\documentclass[floatfix,rmp,%
               amsmath,amssymb,%
               superscriptaddress,%
               twocolumn,nofootinbib]{revtex4-1}

\setcitestyle{numbers,square,comma}

\usepackage{enumitem}
\usepackage{amsthm}
\usepackage{color} 
\usepackage{graphicx}
\usepackage{numprint}
\graphicspath{ {figures/}{figures_SI/} }
\definecolor{color1}{rgb}{0.894118,0.101961,0.109804}
\definecolor{color2}{rgb}{0.215686,0.494118,0.721569}

\newcommand{\reals}{\mathbb{R}}

\newcommand{\exponent}{\operatorname{e}}
\newcommand{\diag}{\operatorname{diag}}

\newcommand{\trace}{\operatorname{tr}}

\newtheoremstyle{sfstyle}% name of the style to be used
  {\topsep}% measure of space to leave above the theorem. E.g.: 3pt
  {\topsep}% measure of space to leave below the theorem. E.g.: 3pt
  {\itshape}% name of font to use in the body of the theorem
  {}% measure of space to indent
  {\sffamily\bfseries\small}% name of head font
  {.}% punctuation between head and body
  {.5em}% space after theorem head; " " = normal interword space
  {}% Manually specify head

\theoremstyle{sfstyle}
\newtheorem{definition}{Definition}
\newtheorem{proposition}{Proposition}
\newtheorem{theorem}{Theorem}

\newtheorem{lemma}{Lemma}
\newtheorem{assumption}{Assumption}

\def \gentitle {Dynamical Models Explaining Social Balance and Evolution of Cooperation}

\begin{document}

  \author{V.A. Traag}
  \affiliation{ICTEAM, Universit\'e catholique de Louvain}
  \author{P. Van Dooren}
  \affiliation{ICTEAM, Universit\'e catholique de Louvain}
  \author{P. De Leenheer} 
  \affiliation{Department of Mathematics, University of Florida}

  \title{\gentitle}

  \begin{abstract}
    Social networks with positive and negative links often split into two
    antagonistic factions. Examples of such a split abound: revolutionaries versus
    an old regime, Republicans versus Democrats, Axis versus Allies during the
    second world war, or the Western versus the Eastern bloc during the Cold War.
    Although this structure, known as social balance, is well understood, it is not
    clear how such factions emerge. An earlier model could explain the formation of
    such factions if reputations were assumed to be symmetric.  We show this is
    not the case for non-symmetric reputations, and propose an alternative model
    which (almost) always leads to social balance, thereby explaining the tendency
    of social networks to split into two factions.  In addition, the alternative
    model may lead to cooperation when faced with defectors, contrary to the earlier
    model. The difference between the two models may be understood in terms of the
    underlying gossiping mechanism: whereas the earlier model assumed that an
    individual adjusts his opinion about somebody by gossiping about that person with
    everybody in the network, we assume instead that the individual gossips with
    that person about everybody. It turns out that the alternative model is
    able to lead to cooperative behaviour, unlike the previous model.
  \end{abstract}

  \maketitle

\section{Introduction}
Why do two antagonistic factions emerge so frequently in social networks?  This
question was already looming in the 1940s, when Heider~\cite{Heider1946}
examined triads of individuals in networks, and postulated  that only balanced
triads are stable. A triad is balanced when friends agree in their opinion of a
third party, while foes disagree, see Fig.~\ref{fig:soc_balance}.  The
individuals in an unbalanced triad have an incentive to adjust their opinions so
as to reduce the stress experienced in such a situation~\cite{Bearman2004}. Once
an adjustment is made, the triad becomes balanced, and the stress disappears. 

A decade later, Harary~\cite{Harary1953} showed that a complete social network
splits in at most two factions if and only if all its triads are balanced, see
also~\cite{Cartwright1956}. Such networks are called (socially) balanced as
well. Since then, the focus of much of the research has been on detecting such
factions in signed networks~\cite{Networks1996,Traag2009}. Many signed networks
show evidence of social balance, although the split into factions might not be
exact, that is, they are only nearly socially
balanced~\cite{Szell2010,Leskovec2010,Facchetti2011,Kunegis2002}.

What has been lacking until fairly recently, are dynamical models that explain
{\it how} social balance emerges. The purpose of this paper is to analyse two
such models. One of these models, proposed first in~\cite{Kulakowski2005}, was
proved to exhibit social balance in~\cite{Marvel2011}. However, this was done
under a restrictive symmetry assumption for  the reputation matrix. Here, we
continue the analysis of this model and show that it   generically does not lead
to social balance when the symmetry assumption is dropped.  In contrast, we
propose a second model that is based on a different underlying gossiping
mechanism, and show that it generically does lead to social balance, even when
reputations are not symmetric.

Moreover, there is a natural connection between negative links and the evolution
of cooperation: we consider positive links as indicating cooperation and
negative links as defection. We will show that our alternative model is able to
lead to cooperation, whereas the earlier model cannot.

\begin{figure}[t]
  \begin{center}
    \includegraphics{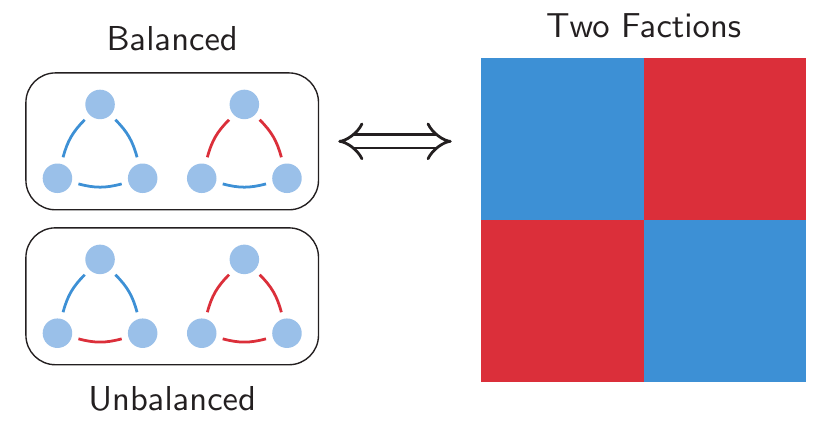}
  \end{center}
  \caption{\textbf{Social Balance.} The two upper triads are balanced, while the
  two lower triads are unbalanced. According to the structure
  theorem~\cite{Harary1953}, a complete graph can be split into (at most) two
  opposing factions, if and only if all triads are balanced. This is represented
  by the coloured matrix on the right, where blue indicates positive entries, and
  red negative entries.\label{fig:soc_balance}}
\end{figure}

\section{Earlier model}

Certain discrete-time, stochastic dynamics have been
investigated~\cite{Antal2005,Radicchi2007}, but they exhibit so-called jammed
states~\cite{Marvel2009}: no change in the sign of a reputation improves the
degree of social balance, as measured by the total number of balanced triads in
the network. A surprisingly simple continuous-time model~\cite{Kulakowski2005}
was proved to converge to social balance for certain symmetric initial
conditions~\cite{Marvel2011}. The authors assume that the social network is
described by a complete graph (everybody is connected to everybody), with
weighted links representing reputations that change continuously in time. Let
$X$ denote the real-valued matrix of the reputations, so that $X_{ij}$
represents the opinion $i$ has about $j$.  It is positive whenever $i$ considers
$j$ a friend, and negative if $i$ thinks of $j$ as an enemy. The network is
balanced,  if, up to a possible relabelling of the individuals, the sign
structure of $X$ takes one of two possible block forms:
\begin{equation}\label{sign-structure}
\begin{pmatrix}
+
\end{pmatrix}\textrm{ or }
\begin{pmatrix}
+&-\\
-&+
\end{pmatrix}.
\end{equation}
Changes in the reputations are modelled as follows:
\begin{equation}
  \dot{X} = X^2,~\textrm{~or~}~\dot{X}_{ij} = \sum_{k} X_{ik} X_{kj},
  \label{equ:X^2}
\end{equation}
where $\dot{X}$ denotes the derivative with respect to time of the matrix $X$.
The idea behind this model is that reputations are adjusted based on the outcome
of a particular gossiping process.  More specifically, suppose that Bob
(individual $i$) wants to revise his opinion about John (individual $j$).  Bob
then asks everybody else in the network what they think of John. If one such
opinion $X_{kj}$ has the same sign as the opinion Bob has about his gossiping
partner, i.e. as $X_{ik}$, then Bob will increase his opinion about John.  But
if these opinions differ in sign, then Bob will decrease his opinion about John.

\begin{figure*}[t]
  \begin{center}
    \includegraphics[width=\textwidth]{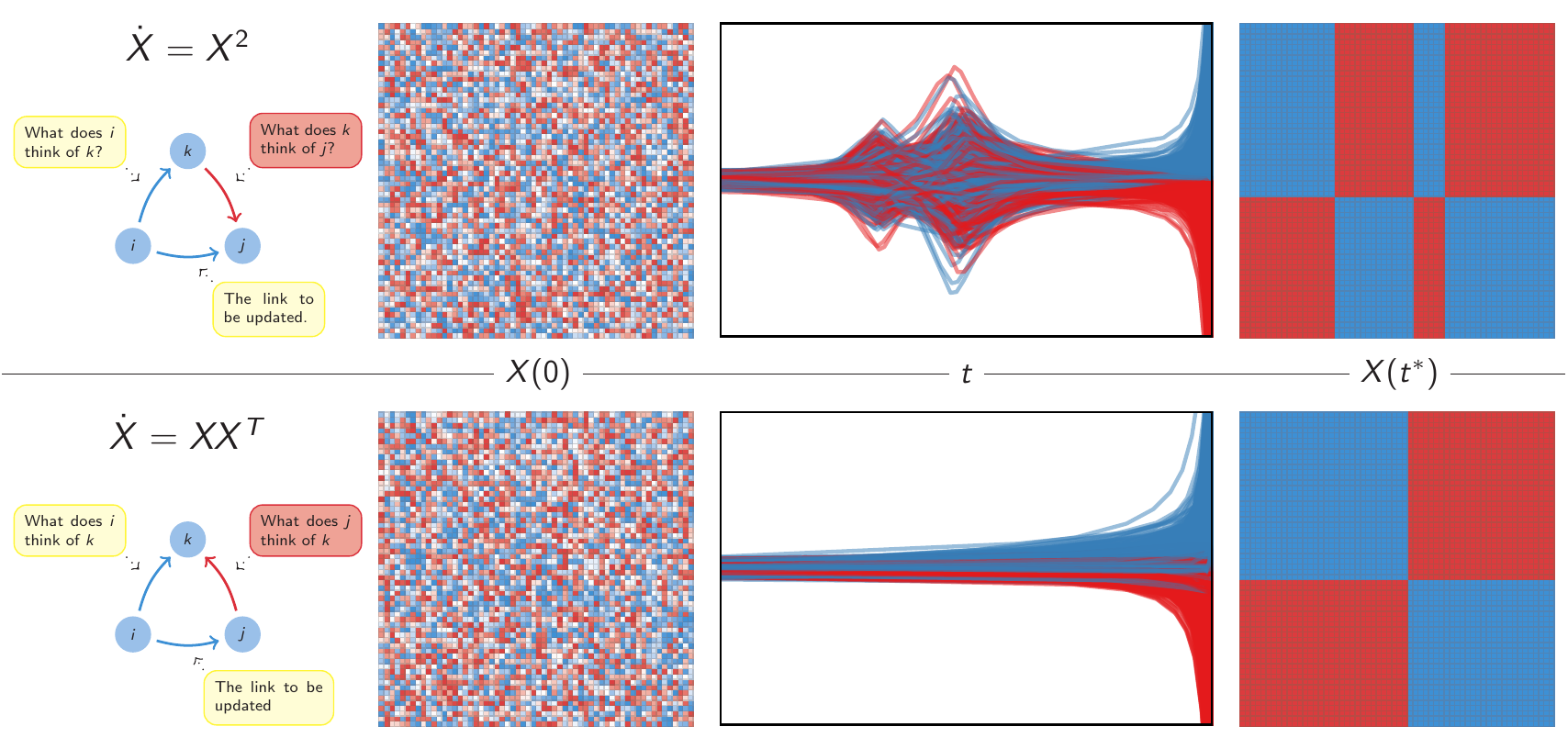}
  \end{center}
  \caption{\textbf{The two models compared.} The first row illustrates what happens
  generically for the model $\dot{X}=X^2$, while the second row displays the
  results for $\dot{X}=XX^T$. Each row contains from left to right: (1) an
  illustration of the model; (2) the random initial state; (3) the dynamics of
  the model; and (4) the final state to which the dynamics converge. Blue
  indicates positive entries, and red negative entries. Although the first model
  converges to a rank one matrix, it is not socially balanced.  The second model
  does converge generically to social balance. The small bumps in the dynamics
  for $\dot{X} = X^2$ are due to complex eigenvalues that show circular
  behaviour (see Fig.~\ref{phase}).\label{fig:illustration}}
\end{figure*}

The analysis for symmetric initial conditions $X(0)=X^T(0)$ was carried out
in~\cite{Marvel2011}: First, $X(0)$ is diagonalized by an orthogonal
transformation $X(0) = U \Lambda(0) U^T$, where the columns of $U$ are
orthonormal  eigenvectors $u_1, \ldots, u_n$ of $X(0)$ so that $UU^T=I_n$, and
$\Lambda(0)$ is a diagonal matrix whose diagonal entries are the corresponding
real eigenvalues $\lambda_1(0) \geq \lambda_2(0) \geq \cdots \geq \lambda_n(0)$
of $X(0)$. Direct substitution of the matrix function $U\Lambda(t)U^T$ shows
that it is the solution of Eq.~\ref{equ:X^2} with initial condition $X(0)$.
Here, $\Lambda(t)$ is a diagonal matrix, solving the uncoupled matrix equation
$\dot{\Lambda} = \Lambda^2$ with initial condition $\Lambda(0)$. The diagonal
entries of $\Lambda(t)$ are obtained by integrating the scalar first order
equations ${\dot \lambda}_i=\lambda_i^2$:
\begin{equation}
  \lambda_i(t) = \frac{\lambda_i(0)}{1 - \lambda_i(0)t},\;\; t\in \begin{cases}[0,+\infty)\textrm{ if }\lambda_i(0)\leq 0\\
[0,1/\lambda_i(0))\textrm{ if } \lambda_i(0)>0
\end{cases}
  \label{equ:real_eig}
\end{equation}
Hence, the solution $X(t)$ blows up in finite time if and only if
$\lambda_1(0)>0$. Moreover, if $\lambda_1(0)>0$ is a simple eigenvalue, then the
solution $X(t)$, normalized by its Frobenius norm, satisfies:
\begin{equation}\label{lim-rank1}
\lim_{t \to 1/\lambda_1(0)}\frac{X(t)}{|X(t)|_F} = u_1 u_1^T.
\end{equation}
Assuming that $u_1$ has no zero entries, and up to a suitable permutation of its
components, the latter limit takes one of the forms in Eq.~\ref{sign-structure}.
In other words, if the initial reputation matrix is symmetric and has a simple,
positive eigenvalue, then the normalized reputation matrix becomes balanced in
finite time.

Our first main result is that this conclusion remains valid for normal initial
conditions, i.e. for initial conditions that satisfy the equality $X(0)X^T(0) =
X^T(0)X(0)$, see SI Theorem~\ref{main-X^2}. Whereas the real eigenvalues
behave similar to the symmetric case, the complex eigenvalues show circular
behaviour, which results in small ``bumps'' in the dynamics as shown in
Fig.~\ref{fig:illustration} (see Fig.~\ref{phase} for more detail).  More
precisely, if $X(0)$ is normal and if $\lambda_1(0)$ is a real, positive and
simple eigenvalue which is larger than every other real eigenvalue (if any),
then the solution $X(t)$ of Eq.~\ref{equ:X^2} satisfies Eq.~\ref{lim-rank1}.
Hence, once again, the normalized reputation matrix converges to a balanced
state.

Our second main result is that this conclusion does not carry over to the case
where $X(0)$ is not normal, see SI Theorem~\ref{main2-X^2}. This general case
is analysed by first transforming $X(0)$ into its real Jordan-canonical form
$J(0)$:  $X(0)=TJ(0)T^{-1}$, where $T$ consists of a basis of (the real and
imaginary parts of) generalized eigenvectors of $X(0)$. It can then be shown
that the solution $X(t)$ of Eq.~\ref{equ:X^2} is given by $TJ(t)T^{-1}$, where
$J(t)$ solves the matrix equation ${\dot J}=J^2$, an equation which can still be
solved explicitly. Hence, $X(t)$ can still be determined. It turns out that if
$X(0)$ has a real, positive and simple eigenvalue $\lambda_1(0)$ which is larger
than every other real eigenvalue (if any), then the normalized reputation matrix
satisfies:
\begin{equation}\label{lim-rank1-no-balance}
\lim_{t \to 1/\lambda_1(0)}\frac{X(t)}{|X(t)|_F} = \frac{u_1 v_1^T}{|u_1v_1^T|_F},
\end{equation}
where $u_1$ and $v_1^T$ are left and right eigenvectors of $X(0)$ respectively,
that correspond to the eigenvalue $\lambda_1(0)$. If we assume that none of the
entries of $u_1$ and $v_1$ are zero, then we can always find a suitable
permutation of the components of $u_1$ and $v_1$ such that they have the
following sign structure:
$$
u_1=\begin{pmatrix}
+\\
+\\
\hline
-\\
-
\end{pmatrix}\textrm{~and~} v_1^T=\begin{pmatrix}
+&-&\vline&+&-
\end{pmatrix}
$$
Consequently, in general, the matrix limit in Eq.~\ref{lim-rank1-no-balance} has the sign structure:
$$
\begin{pmatrix}
+&-&\vline&+&-\\
-&+&\vline&-&+
\end{pmatrix},
$$ 
as illustrated in Fig.~\ref{fig:illustration}. Clearly, this configuration doesn't correspond 
to social balance any longer.

\section{Alternative model}

Let us briefly reconsider the gossiping process underlying model $\dot{X} =
X^2$. In our example of Bob and John, the following happens. Bob asks others
what they think of John. Bob takes into account what he thinks of the people he
talks to, and adjusts his opinion of John accordingly.  An alternative approach
is to consider a type of homophily
process~\cite{McPherson2001,Durrett2005,Fu2012}: people tend to befriend people
who think alike. When Bob seeks to revise his opinion of John, he talks to John
about everybody else (instead of talking to everybody else about John). For
example, suppose that Bob likes Alice, but that John dislikes her.  When Bob and
John talk about Alice, they notice they have opposing views about her, and
as a result the relationship between Bob and John deteriorates. On the other
hand, should they share similar opinions about Alice,  
their relationship will improve. Thus, our alternative model for the update law
of the reputations is:
\begin{equation} \dot{X} =
  XX^T,\textrm{~or~} \dot{X}_{ij} = \sum_k X_{ik}X_{jk}.  \label{equ:XX^T}
\end{equation} 
Although there apparently is only a subtle difference in the gossiping processes
underlying the models in Eq.~\ref{equ:X^2} and \ref{equ:XX^T}, these models turn
out to behave quite differently, as we discuss next.

Our third main result is that for generic initial conditions, the normalized
solution of system Eq.~\ref{equ:XX^T} converges to a socially balanced state in
finite time. To show this, we decompose the solution $X(t)$ into its symmetric
and skew-symmetric parts: $X(t) = S(t) + A(t)$, where $S(t)=S^T(t)$ and
$A(t)=-A^T(t)$. Since $\dot{X} =\dot{X}^T$, the skew-symmetric part remains
constant, and therefore $A(t) = A(0)\equiv A_0$. The symmetric part then obeys
the matrix Riccati differential equation $\dot{S} = (S + A_0)(S - A_0)$.  We
introduce $Z(t) = e^{-A_0t}S(t)e^{A_0t}$ to eliminate the linear terms in this
equation, and obtain 
\begin{align}
  \dot{Z} = %e^{-A_0t} (-A_0 S + \dot{S} + S A_0) e^{A_0t} \\ 
  Z^2 + A_0A_0^T.
\end{align}
The latter matrix Riccati differential equation can be integrated, yielding  the
solution $Z(t)$ explicitly, and hence $S(t)$, as well as $X(t)$, can be
calculated. 

It turns out that if $A_0\neq 0$, then $X(t)$ always blows up in finite time.
Moreover, using a perturbation argument, it can be shown there is a dense set of
initial conditions $X(0)$ such that the normalized solution of
Eq.~\ref{equ:XX^T} converges to 
\begin{equation}\label{lim-rank1-XX^T}
\lim_{t \to t^*}\frac{X(t)}{|X(t)|_F} = w w^T,
\end{equation}
for some vector $w$, as $t$ approaches the blow-up time $t^*$, see SI
Theorem~\ref{thm:generic}. If $w$ has no zero entries, this implies that the
normalized solution becomes balanced in finite time. Hence, the alternative
model in Eq.~\ref{equ:XX^T} generically evolves to social balance, see
Fig.~\ref{fig:illustration}.

\section{Evolution of Cooperation}

\begin{figure}[t]
  \begin{center}
    \includegraphics{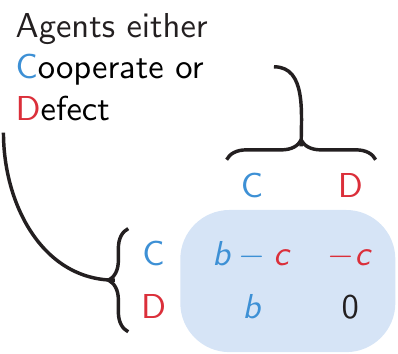}
  \end{center}
  \caption{\textbf{Prisoner's Dilemma.} Both players have the option to either
  Cooperate or Defect. Whenever an agent
  cooperates, it costs him $c$ while his partners receives a benefit $b > c$,
  leading to the indicated payoffs.\label{fig:PD}}
\end{figure}

Positive and negative links have a natural interpretation in the light of
cooperation: positive links indicate cooperation and negative links indicate
defection.  There is then also a natural motivation for the alternative model in
terms of cooperation. Again, suppose Bob wants to revise his opinion of John.
For Bob it is important to know whether John is cooperative in order to
determine whether he should cooperate with John or not. So, instead of asking
Alice whether she has cooperated with John, Bob would like to know whether John
has cooperated with her. In other words, Bob is not interested in $X_{kj}$ but
in $X_{jk}$, consistent with Eq.~\ref{equ:XX^T}, illustrated in
Fig.~\ref{fig:illustration}.  This is also what is observed in studies on
gossip: it often concerns what others did, not what one thinks of
others~\cite{Mcandrew2007,Paine1967}

Indeed gossiping seems crucial in explaining the evolution of human cooperation
through indirect reciprocity~\cite{Nowak2005}. It has even been suggested that
humans developed larger brains in order to gossip, so as to control the problem
of cooperation through social interaction~\cite{Dunbar1998}. In general, the
problem is that if defection allows individuals to gain more, why then do
individuals cooperate? This is usually modelled in the form of a prisoner's
dilemma, in which each agent has the possibility to give his partner some
benefit $b$ at some cost $c<b$. So, if an agent's partner cooperates (he gives
the agent $b$) but the agent doesn't cooperate (he doesn't pay the cost $c$)
his total payoff will be $b$. Considering the other possibilities results in the
payoff matrix detailed in Fig.~\ref{fig:PD}.

Irrespective of the choice of the other player, it is better to defect in a
single game. Suppose that the second player cooperates. Then if the first player
cooperates he gains $b-c$, while if he defects he gains $b$, so defecting is
preferable.  Now suppose that the second player defects. The first player then
has to pay $c$, but doesn't have to pay anything when defecting. So indeed, in a
single game, it is always better to defect, yet the payoff is higher if both
cooperate, whence the dilemma.

In reality, we do observe cooperation, and various mechanisms for explaining the
evolution of cooperation have been suggested~\cite{Nowak2006a}, such as kin
selection~\cite{MaynardSmith1982, Hamilton1964}, reciprocity~\cite{Axelrod1985}
or group selection~\cite{Wilson1975}. Humans have a tendency however to also
cooperate in contexts beyond kin, group or repeated interactions. It is believed
that some form of indirect reciprocity can explain the breadth of human
cooperation~\cite{Nowak2005}. Whereas in direct reciprocity the favour is
returned by the interaction partner, in indirect reciprocity the favour is
returned by somebody else, which usually involves some reputation. It has been
theorized that such a mechanism could even form the basis of
morality~\cite{Alexander1987}. Additionally, reputation (and the fear of losing
reputation) seems to play an important role in maintaining social
norms~\cite{Elias1965, Friedkin2001,Fehr2004}.

In general, the idea is the following: agents obtain some good reputation by
helping others, and others help those with a good reputation. Initially a
strategy known as image scoring was introduced~\cite{Nowak1998}. Shortly after,
it was argued that a different strategy, known as the standing strategy, should
actually perform better~\cite{Leimar2001}, although experiments showed people
tend to prefer the simpler image scoring strategy~\cite{Milinski2001}. This led
to more systematic studies of how different reputation schemes would
perform~\cite{A2006a,Brandt2004,A2004}. Although much research has been done on
indirect reciprocity, only few theoretical works actually study how gossiping
shapes reputations~\cite{Nakamaru2004,Traag2011a}. Nonetheless, most studies
(tacitly) assume that reputations are shaped through gossip. Additionally, it
was observed experimentally that gossiping is an effective mechanism for
promoting cooperation~\cite{Piazza2008,Sommerfeld2009,Semmann2007}.

Moreover, these reputations are usually considered as objective. That is, all
agents know the reputation $X_j$ of some agent $j$, and all agents have the same
view of agent $j$. Private reputations---so that we have $X_{ij}$, the
reputation of $j$ in the eyes of $i$---have usually been considered by allowing
a part of the population to ``observe'' an interaction, and update the
reputation accordingly.  If too few agents are allowed to ``observe'' an
interaction, the reputations $X_{ij}$ tend to become uncorrelated and
incoherent. This makes reputation unreliable for deciding whether to cooperate
or defect. The central question thus becomes how to model private reputations
such that they remain coherent and reliable for deciding whether to cooperate or
not.

Dynamical models of social balance might provide an answer to this question.
Although it allows to have private reputations---that is $X_{ij}$---the dynamics
could also lead to some coherence in the form of social balance. In addition, it
models more explicitly the gossiping process, commonly suggested to be the
foundation upon which reputations are forged. 

\begin{figure*}[t]
  \begin{center}
    \includegraphics{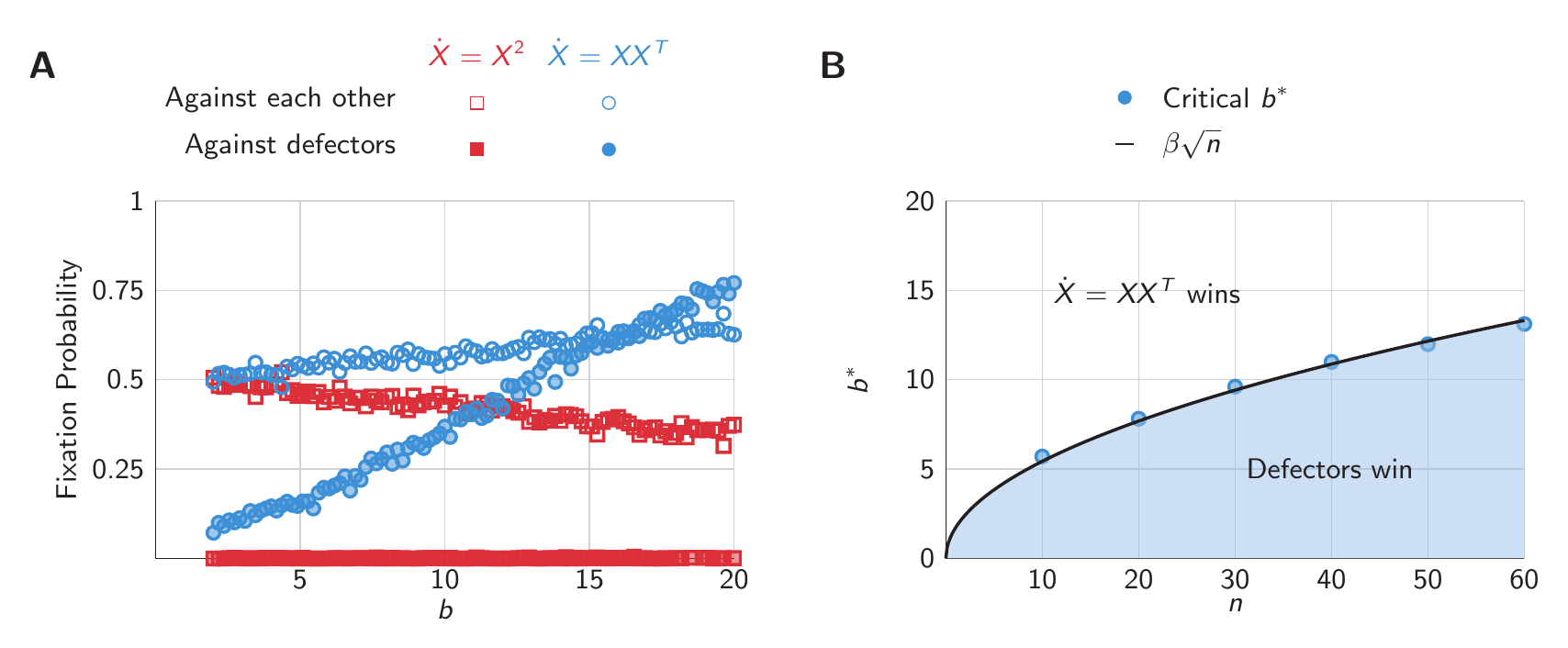}
  \end{center}
  \caption{\textbf{Evolution of Cooperation.} (A) The fixation probability
  (probability to be the sole surviving species) is higher for model
  $\dot{X}=XX^T$ than $\dot{X}=X^2$. This implies that the model $\dot{X}=XX^T$
  is more viable against defectors, and has an evolutionary advantage compared
  to $\dot{X}=X^2$. (B) The point $b^*$ at which the model $\dot{X} = XX^T$ has
  an evolutionary advantage against defectors (i.e. the fixation probability
  $\rho > 1/2$) depends on the number of agents $n$. The condition for the model
  $\dot{X} = XX^T$ to defeat defectors can be approximated by $b > b^* = \beta
  \sqrt{n}$, with $\beta \approx 1.72$.\label{fig:evol}}
\end{figure*}

\subsection{Simulation Results}

The reputations of the agents are determined by the dynamics of the two models.
We call agents using $\dot{X} = X^2$ dynamics type A, and those using $\dot{X} =
X X^T$ dynamics type B. We assume that agent $i$ cooperates with $j$ whenever
$X_{ij} > 0$ and defects otherwise. Agents reproduce proportional to their
fitness, determined by their payoff. Agents that do well (have a high payoff)
have a higher chance of reproduction, and we are interested in knowing the
probability that a certain type becomes fixated in the population (i.e.  takes
over the whole population), known as the fixation probability $\rho$. All
simulations start off with an equal amount of agents, so if some type wins more
often than his initial relative frequency, it indicates it has an evolutionary
advantage. For the results presented here this comes down to $\rho > 1/2$. More
details on the simulations are provided in the Materials and Methods section at
the end of the paper. 

The results are displayed in Fig.~\ref{fig:evol} using a normalized cost of
$c=1$ (the ratio $b/c$ drives the evolutionary dynamics, see Materials and
Methods and~\cite{Nowak2006a}). When directly competing against each other, type
B has an evolutionary advantage (its fixation probability $\rho_B > 1/2$)
compared to type A, already for relatively small benefits. When each type is
playing against defectors (agents that always defect), type A seems unable to
defeat defectors ($\rho_A < 1/2$) for any $b < 20$, while type B performs quite
well against them. When all three types are playing against each other results
are similar (see Fig.~\ref{fig:defect_ABD}). When varying the number of
agents, the critical benefit $b^*$ at which type B starts to have an
evolutionary advantage changes (i.e. where the fixation probability $\rho_B =
1/2$). For $b > b^*$ agents using the model $\dot{X} = XX^T$ have a higher
chance to become fixated, while for $b < b^*$ defectors tend to win. The
inequality for type B to have an evolutionary advantage can be relatively
accurately approximated by $b > b^* = \gamma \sqrt{n}$ where $\gamma$ is estimated
to be around $\gamma \approx 1.72 \pm 0.037$ (95\% confidence interval). Varying
the intensity of selection does not alter the results qualitatively (see
Fig.~\ref{fig:beta}).  Summarizing, type B is able to lead to cooperation and
defeats type A. Based on these results, if a gossiping process evolved during
the course of human history in order to maintain cooperation, the model $\dot{X}
= XX^T$ seems more likely to have evolved than $\dot{X} = X^2$.  For smaller
groups a smaller benefit is needed for the model $\dot{X} = XX^T$ to become
fixated. This dependence seems to scale only as $\sqrt{n}$, so that larger
groups only need a marginally larger benefit in order to develop cooperation.

\section{Conclusion}

To conclude, we have shown that the alternative model $\dot{X} = XX^T$ generically 
converges to social balance, whereas the model $\dot{X} = X^2$ did not.
The current models exhibit several unrealistic features, we would like to
address: (1) an all-to-all topology; (2) dynamics that blow-up in finite time;
and (3) homogeneity of all agents. Although most of these issues can be
addressed by specifying different dynamics, the resulting models are much more
difficult to analyse, thereby limiting our understanding. Although the two
models are somewhat simple, they are also tractable, and what we lose in
truthfulness, we gain in deeper insights: in simplicity lies progress. Our
current analysis offers a quite complete understanding, and we hope it provides
a stepping stone to more realistic models, which we would like to analyse in the
future.

The difference between the two models can be understood in terms of gossiping:
we assume that people who wish to revise their opinion about someone talk to
that person about everybody else, while the earlier model assumed that people
talk about that person to everybody else. Both gossiping and social balance are
at the centre of many social
phenomena~\cite{Gluckman1963,Elias1965,Fostera,Dunbar1998}, such as norm
maintenance~\cite{Friedkin2001}, stereotype formation~\cite{Wert} and social
conflict~\cite{Labianca1998}. For example, a classic work~\cite{Elias1965} on
the established and outsiders found that gossiping was the fundamental driving
force for the maintenance of the cohesive network of the established at the
exclusion of the outsiders. Understanding how social balance may emerge might
help to understand the intricacies of these social phenomena.

Moreover, in light of the evolution of cooperation it appears that agents using
$\dot{X} = XX^T$ dynamics perform well against defectors, and have an
evolutionary advantage compared to agents using $\dot{X} = X^2$ dynamics.
Contrary to other models of indirect reciprocity, not everybody might end up
cooperating with everybody, and the population may split into two groups. This
provides an interesting connection between social balance theory, gossiping and
the evolution of cooperation. Our results improve our understanding of gossiping
as a mechanism for group formation and cooperation, and as such contributes to
the study of indirect reciprocity.

\begin{acknowledgments}

    We acknowledge support from a grant ``Actions de recherche concert\'ees
    --- Large Graphs and Networks'' of the ``Communaut\'e Fran\c caise de
    Belgique'' and from the Belgian Network DYSCO (Dynamical Systems, Control,
    and Optimization), funded by the Interuniversity Attraction Poles
    Programme, initiated by the Belgian State, Science Policy Office. P.D.L.
    wishes to thank VLAC (Vlaams Academisch Centrum) for providing him with a
    research fellowship during a sabbatical leave from the University of
    Florida, the Universit{\'e} catholique de Louvain for hosting him as a
    visiting professor, and the University of Florida for a Faculty
    Enhancement Opportunity.

\end{acknowledgments}

\bibliography{bibliography}

%merlin.mbs apsrmp4-1.bst 2010-07-25 4.21a (PWD, AO, DPC) hacked
%Control: author (11) reversed first initials
\begin{thebibliography}{54}%
\makeatletter
\providecommand \@ifxundefined [1]{%
 \@ifx{#1\undefined}
}%
\providecommand \@ifnum [1]{%
 \ifnum #1\expandafter \@firstoftwo
 \else \expandafter \@secondoftwo
 \fi
}%
\providecommand \@ifx [1]{%
 \ifx #1\expandafter \@firstoftwo
 \else \expandafter \@secondoftwo
 \fi
}%
\providecommand \natexlab [1]{#1}%
\providecommand \enquote  [1]{``#1''}%
\providecommand \bibnamefont  [1]{#1}%
\providecommand \bibfnamefont [1]{#1}%
\providecommand \citenamefont [1]{#1}%
\providecommand \href@noop [0]{\@secondoftwo}%
\providecommand \href [0]{\begingroup \@sanitize@url \@href}%
\providecommand \@href[1]{\@@startlink{#1}\@@href}%
\providecommand \@@href[1]{\endgroup#1\@@endlink}%
\providecommand \@sanitize@url [0]{\catcode `\\12\catcode `\$12\catcode
  `\&12\catcode `\#12\catcode `\^12\catcode `\_12\catcode `\%12\relax}%
\providecommand \@@startlink[1]{}%
\providecommand \@@endlink[0]{}%
\providecommand \url  [0]{\begingroup\@sanitize@url \@url }%
\providecommand \@url [1]{\endgroup\@href {#1}{\urlprefix }}%
\providecommand \urlprefix  [0]{URL }%
\providecommand \Eprint [0]{\href }%
\providecommand \doibase [0]{http://dx.doi.org/}%
\providecommand \selectlanguage [0]{\@gobble}%
\providecommand \bibinfo  [0]{\@secondoftwo}%
\providecommand \bibfield  [0]{\@secondoftwo}%
\providecommand \translation [1]{[#1]}%
\providecommand \BibitemOpen [0]{}%
\providecommand \bibitemStop [0]{}%
\providecommand \bibitemNoStop [0]{.\EOS\space}%
\providecommand \EOS [0]{\spacefactor3000\relax}%
\providecommand \BibitemShut  [1]{\csname bibitem#1\endcsname}%
\let\auto@bib@innerbib\@empty
%</preamble>
\bibitem [{\citenamefont {Alexander}(1987)}]{Alexander1987}%
  \BibitemOpen
  \bibfield  {author} {\bibinfo {author} {\bibnamefont {Alexander},
  \bibfnamefont {R.~D.}}} (\bibinfo {year} {1987}),\ \href@noop {} {\emph
  {\bibinfo {title} {{The Biology of Moral Systems}}}}\ (\bibinfo  {publisher}
  {Aldine de Gruyter},\ \bibinfo {address} {New York})\BibitemShut {NoStop}%
\bibitem [{\citenamefont {Antal}\ \emph {et~al.}(2005)\citenamefont {Antal},
  \citenamefont {Krapivsky},\ and\ \citenamefont {Redner}}]{Antal2005}%
  \BibitemOpen
  \bibfield  {author} {\bibinfo {author} {\bibnamefont {Antal}, \bibfnamefont
  {T.}}, \bibinfo {author} {\bibfnamefont {P.~L.}\ \bibnamefont {Krapivsky}}, \
  and\ \bibinfo {author} {\bibfnamefont {S.}~\bibnamefont {Redner}}} (\bibinfo
  {year} {2005}),\ \href@noop {} {\bibfield  {journal} {\bibinfo  {journal}
  {Phys. Rev. E}\ }\textbf {\bibinfo {volume} {72}}~(\bibinfo {number} {3}),\
  \bibinfo {pages} {36121}}\BibitemShut {NoStop}%
\bibitem [{\citenamefont {Axelrod}\ and\ \citenamefont
  {Hamilton}(1985)}]{Axelrod1985}%
  \BibitemOpen
  \bibfield  {author} {\bibinfo {author} {\bibnamefont {Axelrod}, \bibfnamefont
  {R.}}, \ and\ \bibinfo {author} {\bibfnamefont {W.}~\bibnamefont {Hamilton}}}
  (\bibinfo {year} {1985}),\ \href@noop {} {\emph {\bibinfo {title}
  {Science}}},\ Vol.\ \bibinfo {volume} {211}\ (\bibinfo  {publisher} {Basic
  Books},\ \bibinfo {address} {New York})\BibitemShut {NoStop}%
\bibitem [{\citenamefont {Bearman}\ and\ \citenamefont
  {Moody}(2004)}]{Bearman2004}%
  \BibitemOpen
  \bibfield  {author} {\bibinfo {author} {\bibnamefont {Bearman}, \bibfnamefont
  {P.~S.}}, \ and\ \bibinfo {author} {\bibfnamefont {J.}~\bibnamefont {Moody}}}
  (\bibinfo {year} {2004}),\ \href@noop {} {\bibfield  {journal} {\bibinfo
  {journal} {Am J Public Health}\ }\textbf {\bibinfo {volume} {94}}~(\bibinfo
  {number} {1}),\ \bibinfo {pages} {89}}\BibitemShut {NoStop}%
\bibitem [{\citenamefont {Brandt}\ and\ \citenamefont
  {Sigmund}(2004)}]{Brandt2004}%
  \BibitemOpen
  \bibfield  {author} {\bibinfo {author} {\bibnamefont {Brandt}, \bibfnamefont
  {H.}}, \ and\ \bibinfo {author} {\bibfnamefont {K.}~\bibnamefont {Sigmund}}}
  (\bibinfo {year} {2004}),\ \href@noop {} {\bibfield  {journal} {\bibinfo
  {journal} {J Theor Biol}\ }\textbf {\bibinfo {volume} {231}}~(\bibinfo
  {number} {4}),\ \bibinfo {pages} {475}}\BibitemShut {NoStop}%
\bibitem [{\citenamefont {Bunse-Gerstner}\ \emph {et~al.}(1991)\citenamefont
  {Bunse-Gerstner}, \citenamefont {Byers}, \citenamefont {Mehrmann},\ and\
  \citenamefont {Nichols}}]{Bunse-Gerstner1991}%
  \BibitemOpen
  \bibfield  {author} {\bibinfo {author} {\bibnamefont {Bunse-Gerstner},
  \bibfnamefont {A.}}, \bibinfo {author} {\bibfnamefont {R.}~\bibnamefont
  {Byers}}, \bibinfo {author} {\bibfnamefont {V.}~\bibnamefont {Mehrmann}}, \
  and\ \bibinfo {author} {\bibfnamefont {N.~K.}\ \bibnamefont {Nichols}}}
  (\bibinfo {year} {1991}),\ \href {\doibase 10.1007/BF01385712} {\bibfield
  {journal} {\bibinfo  {journal} {Numerische Mathematik}\ }\textbf {\bibinfo
  {volume} {60}}~(\bibinfo {number} {1}),\ \bibinfo {pages} {1}}\BibitemShut
  {NoStop}%
\bibitem [{\citenamefont {Cartwright}\ and\ \citenamefont
  {Harary}(1956{\natexlab{a}})}]{Cartwright1956}%
  \BibitemOpen
  \bibfield  {author} {\bibinfo {author} {\bibnamefont {Cartwright},
  \bibfnamefont {D.}}, \ and\ \bibinfo {author} {\bibfnamefont
  {F.}~\bibnamefont {Harary}}} (\bibinfo {year} {1956}{\natexlab{a}}),\
  \href@noop {} {\bibfield  {journal} {\bibinfo  {journal} {Psychol Rev}\
  }\textbf {\bibinfo {volume} {63}}~(\bibinfo {number} {5}),\ \bibinfo {pages}
  {277}}\BibitemShut {NoStop}%
\bibitem [{\citenamefont {Cartwright}\ and\ \citenamefont
  {Harary}(1956{\natexlab{b}})}]{Cartwrighta}%
  \BibitemOpen
  \bibfield  {author} {\bibinfo {author} {\bibnamefont {Cartwright},
  \bibfnamefont {D.}}, \ and\ \bibinfo {author} {\bibfnamefont
  {F.}~\bibnamefont {Harary}}} (\bibinfo {year} {1956}{\natexlab{b}}),\ \href
  {\doibase 10.1037/h0046049} {\bibfield  {journal} {\bibinfo  {journal}
  {Psychological Review}\ }\textbf {\bibinfo {volume} {63}}~(\bibinfo {number}
  {5}),\ \bibinfo {pages} {277}}\BibitemShut {NoStop}%
\bibitem [{\citenamefont {{De Leenheer}}\ and\ \citenamefont
  {Sontag}(2004)}]{DeLeenheer2004}%
  \BibitemOpen
  \bibfield  {author} {\bibinfo {author} {\bibnamefont {{De Leenheer}},
  \bibfnamefont {P.}}, \ and\ \bibinfo {author} {\bibfnamefont {E.~D.}\
  \bibnamefont {Sontag}}} (\bibinfo {year} {2004}),\ \href@noop {} {\emph
  {\bibinfo {title} {{A note on the monotonicity of matrix Riccati
  equations}}}},\ \bibinfo {type} {Tech. Rep.}\ (\bibinfo  {institution}
  {DIMACS})\BibitemShut {NoStop}%
\bibitem [{\citenamefont {Doreian}\ and\ \citenamefont
  {Mrvar}(1996)}]{Networks1996}%
  \BibitemOpen
  \bibfield  {author} {\bibinfo {author} {\bibnamefont {Doreian}, \bibfnamefont
  {P.}}, \ and\ \bibinfo {author} {\bibfnamefont {A.}~\bibnamefont {Mrvar}}}
  (\bibinfo {year} {1996}),\ \href@noop {} {\bibfield  {journal} {\bibinfo
  {journal} {Soc Networks}\ }\textbf {\bibinfo {volume} {18}}~(\bibinfo
  {number} {2}),\ \bibinfo {pages} {149}}\BibitemShut {NoStop}%
\bibitem [{\citenamefont {Dunbar}(1998)}]{Dunbar1998}%
  \BibitemOpen
  \bibfield  {author} {\bibinfo {author} {\bibnamefont {Dunbar}, \bibfnamefont
  {R.~I.~M.}}} (\bibinfo {year} {1998}),\ \href@noop {} {\emph {\bibinfo
  {title} {{Grooming, Gossip, and the Evolution of Language}}}}\ (\bibinfo
  {publisher} {Harvard University Press},\ \bibinfo {address}
  {Cambridge})\BibitemShut {NoStop}%
\bibitem [{\citenamefont {Durrett}\ and\ \citenamefont
  {Levin}(2005)}]{Durrett2005}%
  \BibitemOpen
  \bibfield  {author} {\bibinfo {author} {\bibnamefont {Durrett}, \bibfnamefont
  {R.}}, \ and\ \bibinfo {author} {\bibfnamefont {S.~A.}\ \bibnamefont
  {Levin}}} (\bibinfo {year} {2005}),\ \href@noop {} {\bibfield  {journal}
  {\bibinfo  {journal} {J Econ Behav Organ}\ }\textbf {\bibinfo {volume}
  {57}}~(\bibinfo {number} {3}),\ \bibinfo {pages} {267}}\BibitemShut {NoStop}%
\bibitem [{\citenamefont {Elias}\ and\ \citenamefont
  {Scotson}(1994)}]{Elias1965}%
  \BibitemOpen
  \bibfield  {author} {\bibinfo {author} {\bibnamefont {Elias}, \bibfnamefont
  {N.}}, \ and\ \bibinfo {author} {\bibfnamefont {J.~L.}\ \bibnamefont
  {Scotson}}} (\bibinfo {year} {1994}),\ \href@noop {} {\emph {\bibinfo {title}
  {{The Established and the Outsiders}}}}\ (\bibinfo  {publisher} {SAGE
  Publications},\ \bibinfo {address} {London})\BibitemShut {NoStop}%
\bibitem [{\citenamefont {Facchetti}\ \emph {et~al.}(2011)\citenamefont
  {Facchetti}, \citenamefont {Iacono},\ and\ \citenamefont
  {Altafini}}]{Facchetti2011}%
  \BibitemOpen
  \bibfield  {author} {\bibinfo {author} {\bibnamefont {Facchetti},
  \bibfnamefont {G.}}, \bibinfo {author} {\bibfnamefont {G.}~\bibnamefont
  {Iacono}}, \ and\ \bibinfo {author} {\bibfnamefont {C.}~\bibnamefont
  {Altafini}}} (\bibinfo {year} {2011}),\ \href@noop {} {\bibfield  {journal}
  {\bibinfo  {journal} {Proc Natl Acad Sci U S A}\ }\textbf {\bibinfo {volume}
  {108}}~(\bibinfo {number} {52}),\ \bibinfo {pages} {20953}}\BibitemShut
  {NoStop}%
\bibitem [{\citenamefont {Fehr}\ and\ \citenamefont
  {Fischbacher}(2004)}]{Fehr2004}%
  \BibitemOpen
  \bibfield  {author} {\bibinfo {author} {\bibnamefont {Fehr}, \bibfnamefont
  {E.}}, \ and\ \bibinfo {author} {\bibfnamefont {U.}~\bibnamefont
  {Fischbacher}}} (\bibinfo {year} {2004}),\ \href@noop {} {\bibfield
  {journal} {\bibinfo  {journal} {Evolution and Human Behavior}\ }\textbf
  {\bibinfo {volume} {25}}~(\bibinfo {number} {2}),\ \bibinfo {pages}
  {63}}\BibitemShut {NoStop}%
\bibitem [{\citenamefont {Foster}(2004)}]{Fostera}%
  \BibitemOpen
  \bibfield  {author} {\bibinfo {author} {\bibnamefont {Foster}, \bibfnamefont
  {E.~K.}}} (\bibinfo {year} {2004}),\ \href@noop {} {\bibfield  {journal}
  {\bibinfo  {journal} {Rev Gen Psychol}\ }\textbf {\bibinfo {volume}
  {8}}~(\bibinfo {number} {2}),\ \bibinfo {pages} {78}}\BibitemShut {NoStop}%
\bibitem [{\citenamefont {Friedkin}(2001)}]{Friedkin2001}%
  \BibitemOpen
  \bibfield  {author} {\bibinfo {author} {\bibnamefont {Friedkin},
  \bibfnamefont {N.~E.}}} (\bibinfo {year} {2001}),\ \href@noop {} {\bibfield
  {journal} {\bibinfo  {journal} {Soc Networks}\ }\textbf {\bibinfo {volume}
  {23}}~(\bibinfo {number} {3}),\ \bibinfo {pages} {167}}\BibitemShut {NoStop}%
\bibitem [{\citenamefont {Fu}\ \emph {et~al.}(2012)\citenamefont {Fu},
  \citenamefont {Nowak}, \citenamefont {Christakis},\ and\ \citenamefont
  {Fowler}}]{Fu2012}%
  \BibitemOpen
  \bibfield  {author} {\bibinfo {author} {\bibnamefont {Fu}, \bibfnamefont
  {F.}}, \bibinfo {author} {\bibfnamefont {M.~A.}\ \bibnamefont {Nowak}},
  \bibinfo {author} {\bibfnamefont {N.~A.}\ \bibnamefont {Christakis}}, \ and\
  \bibinfo {author} {\bibfnamefont {J.~H.}\ \bibnamefont {Fowler}}} (\bibinfo
  {year} {2012}),\ \href@noop {} {\bibfield  {journal} {\bibinfo  {journal}
  {Scientific Reports}\ }\textbf {\bibinfo {volume} {2}}}\BibitemShut {NoStop}%
\bibitem [{\citenamefont {Gillespie}(2004)}]{Gillespie2004}%
  \BibitemOpen
  \bibfield  {author} {\bibinfo {author} {\bibnamefont {Gillespie},
  \bibfnamefont {J.~H.}}} (\bibinfo {year} {2004}),\ \href@noop {} {\emph
  {\bibinfo {title} {{Population Genetics: A Concise Guide}}}}\ (\bibinfo
  {publisher} {The John Hopkins University Press},\ \bibinfo {address}
  {Baltimore})\BibitemShut {NoStop}%
\bibitem [{\citenamefont {Gluckman}(1963)}]{Gluckman1963}%
  \BibitemOpen
  \bibfield  {author} {\bibinfo {author} {\bibnamefont {Gluckman},
  \bibfnamefont {M.}}} (\bibinfo {year} {1963}),\ \href@noop {} {\bibfield
  {journal} {\bibinfo  {journal} {Curr Anthropol}\ }\textbf {\bibinfo {volume}
  {4}}~(\bibinfo {number} {3}),\ \bibinfo {pages} {307}}\BibitemShut {NoStop}%
\bibitem [{\citenamefont {Hamilton}(1964)}]{Hamilton1964}%
  \BibitemOpen
  \bibfield  {author} {\bibinfo {author} {\bibnamefont {Hamilton},
  \bibfnamefont {W.}}} (\bibinfo {year} {1964}),\ \href@noop {} {\bibfield
  {journal} {\bibinfo  {journal} {J Theor Biol}\ }\textbf {\bibinfo {volume}
  {7}}~(\bibinfo {number} {1}),\ \bibinfo {pages} {1}}\BibitemShut {NoStop}%
\bibitem [{\citenamefont {Harary}(1953)}]{Harary1953}%
  \BibitemOpen
  \bibfield  {author} {\bibinfo {author} {\bibnamefont {Harary}, \bibfnamefont
  {F.}}} (\bibinfo {year} {1953}),\ \href@noop {} {\bibfield  {journal}
  {\bibinfo  {journal} {The Michigan Mathematical Journal}\ }\textbf {\bibinfo
  {volume} {2}}~(\bibinfo {number} {2}),\ \bibinfo {pages} {143}}\BibitemShut
  {NoStop}%
\bibitem [{\citenamefont {Heider}(1946)}]{Heider1946}%
  \BibitemOpen
  \bibfield  {author} {\bibinfo {author} {\bibnamefont {Heider}, \bibfnamefont
  {F.}}} (\bibinfo {year} {1946}),\ \href@noop {} {\bibfield  {journal}
  {\bibinfo  {journal} {J Psychol}\ }\textbf {\bibinfo {volume} {21}}~(\bibinfo
  {number} {1}),\ \bibinfo {pages} {107}}\BibitemShut {NoStop}%
\bibitem [{\citenamefont {Horn}\ and\ \citenamefont
  {Johnson}(1985)}]{Horn1985}%
  \BibitemOpen
  \bibfield  {author} {\bibinfo {author} {\bibnamefont {Horn}, \bibfnamefont
  {R.}}, \ and\ \bibinfo {author} {\bibfnamefont {C.}~\bibnamefont {Johnson}}}
  (\bibinfo {year} {1985}),\ \href@noop {} {\emph {\bibinfo {title} {{Matrix
  Analysis}}}}\ (\bibinfo  {publisher} {Cambridge University Press},\ \bibinfo
  {address} {Cambridge})\BibitemShut {NoStop}%
\bibitem [{\citenamefont {Kato}(1995)}]{Kato1995}%
  \BibitemOpen
  \bibfield  {author} {\bibinfo {author} {\bibnamefont {Kato}, \bibfnamefont
  {T.}}} (\bibinfo {year} {1995}),\ \href
  {http://www.amazon.com/Perturbation-Theory-Operators-Classics-Mathematics/dp/354058661X}
  {\emph {\bibinfo {title} {{Perturbation Theory for Linear Operators}}}}\
  (\bibinfo  {publisher} {Springer},\ \bibinfo {address} {New
  York})\BibitemShut {NoStop}%
\bibitem [{\citenamefont {Kulakowski}\ \emph {et~al.}(2005)\citenamefont
  {Kulakowski}, \citenamefont {Gawronski},\ and\ \citenamefont
  {Gronek}}]{Kulakowski2005}%
  \BibitemOpen
  \bibfield  {author} {\bibinfo {author} {\bibnamefont {Kulakowski},
  \bibfnamefont {K.}}, \bibinfo {author} {\bibfnamefont {P.}~\bibnamefont
  {Gawronski}}, \ and\ \bibinfo {author} {\bibfnamefont {P.}~\bibnamefont
  {Gronek}}} (\bibinfo {year} {2005}),\ \href@noop {} {\bibfield  {journal}
  {\bibinfo  {journal} {Int J Mod Phys C}\ }\textbf {\bibinfo {volume}
  {16}}~(\bibinfo {number} {5}),\ \bibinfo {pages} {707}}\BibitemShut {NoStop}%
\bibitem [{\citenamefont {Kunegis}\ \emph {et~al.}(2009)\citenamefont
  {Kunegis}, \citenamefont {Lommatzsch},\ and\ \citenamefont
  {Bauckhage}}]{Kunegis2002}%
  \BibitemOpen
  \bibfield  {author} {\bibinfo {author} {\bibnamefont {Kunegis}, \bibfnamefont
  {J.}}, \bibinfo {author} {\bibfnamefont {A.}~\bibnamefont {Lommatzsch}}, \
  and\ \bibinfo {author} {\bibfnamefont {C.}~\bibnamefont {Bauckhage}}}
  (\bibinfo {year} {2009}),\ in\ \href@noop {} {\emph {\bibinfo {booktitle}
  {Proceedings of the 18th international conference on World wide web - WWW
  '09}}}\ (\bibinfo  {publisher} {ACM Press},\ \bibinfo {address} {New York,
  New York, USA})\ p.\ \bibinfo {pages} {741}\BibitemShut {NoStop}%
\bibitem [{\citenamefont {Labianca}\ \emph {et~al.}(1998)\citenamefont
  {Labianca}, \citenamefont {Brass},\ and\ \citenamefont
  {Gray}}]{Labianca1998}%
  \BibitemOpen
  \bibfield  {author} {\bibinfo {author} {\bibnamefont {Labianca},
  \bibfnamefont {G.}}, \bibinfo {author} {\bibfnamefont {D.}~\bibnamefont
  {Brass}}, \ and\ \bibinfo {author} {\bibfnamefont {B.}~\bibnamefont {Gray}}}
  (\bibinfo {year} {1998}),\ \href@noop {} {\bibfield  {journal} {\bibinfo
  {journal} {Academy of Management journal}\ }\textbf {\bibinfo {volume}
  {41}}~(\bibinfo {number} {1}),\ \bibinfo {pages} {55}}\BibitemShut {NoStop}%
\bibitem [{\citenamefont {Leimar}\ and\ \citenamefont
  {Hammerstein}(2001)}]{Leimar2001}%
  \BibitemOpen
  \bibfield  {author} {\bibinfo {author} {\bibnamefont {Leimar}, \bibfnamefont
  {O.}}, \ and\ \bibinfo {author} {\bibfnamefont {P.}~\bibnamefont
  {Hammerstein}}} (\bibinfo {year} {2001}),\ \href@noop {} {\bibfield
  {journal} {\bibinfo  {journal} {Proc Biol Sci}\ }\textbf {\bibinfo {volume}
  {268}}~(\bibinfo {number} {1468}),\ \bibinfo {pages} {745}}\BibitemShut
  {NoStop}%
\bibitem [{\citenamefont {Leskovec}\ \emph {et~al.}(2010)\citenamefont
  {Leskovec}, \citenamefont {Huttenlocher},\ and\ \citenamefont
  {Kleinberg}}]{Leskovec2010}%
  \BibitemOpen
  \bibfield  {author} {\bibinfo {author} {\bibnamefont {Leskovec},
  \bibfnamefont {J.}}, \bibinfo {author} {\bibfnamefont {D.}~\bibnamefont
  {Huttenlocher}}, \ and\ \bibinfo {author} {\bibfnamefont {J.}~\bibnamefont
  {Kleinberg}}} (\bibinfo {year} {2010}),\ in\ \href@noop {} {\emph {\bibinfo
  {booktitle} {WWW 2010}}}\BibitemShut {NoStop}%
\bibitem [{\citenamefont {Marvel}\ \emph {et~al.}(2009)\citenamefont {Marvel},
  \citenamefont {Strogatz},\ and\ \citenamefont {Kleinberg}}]{Marvel2009}%
  \BibitemOpen
  \bibfield  {author} {\bibinfo {author} {\bibnamefont {Marvel}, \bibfnamefont
  {S.}}, \bibinfo {author} {\bibfnamefont {S.}~\bibnamefont {Strogatz}}, \ and\
  \bibinfo {author} {\bibfnamefont {J.}~\bibnamefont {Kleinberg}}} (\bibinfo
  {year} {2009}),\ \href@noop {} {\bibfield  {journal} {\bibinfo  {journal}
  {Phys Rev Lett}\ }\textbf {\bibinfo {volume} {103}}~(\bibinfo {number}
  {19}),\ \bibinfo {pages} {198701}},\ \Eprint {http://arxiv.org/abs/0906.2893}
  {arXiv:0906.2893} \BibitemShut {NoStop}%
\bibitem [{\citenamefont {Marvel}\ \emph {et~al.}(2011)\citenamefont {Marvel},
  \citenamefont {Kleinberg}, \citenamefont {Kleinberg},\ and\ \citenamefont
  {Strogatz}}]{Marvel2011}%
  \BibitemOpen
  \bibfield  {author} {\bibinfo {author} {\bibnamefont {Marvel}, \bibfnamefont
  {S.~A.}}, \bibinfo {author} {\bibfnamefont {J.}~\bibnamefont {Kleinberg}},
  \bibinfo {author} {\bibfnamefont {R.~D.}\ \bibnamefont {Kleinberg}}, \ and\
  \bibinfo {author} {\bibfnamefont {S.~H.}\ \bibnamefont {Strogatz}}} (\bibinfo
  {year} {2011}),\ \href@noop {} {\bibfield  {journal} {\bibinfo  {journal}
  {Proc Natl Acad Sci U S A}\ }\textbf {\bibinfo {volume} {108}}~(\bibinfo
  {number} {5}),\ \bibinfo {pages} {1771}}\BibitemShut {NoStop}%
\bibitem [{\citenamefont {McAndrew}\ \emph {et~al.}(2007)\citenamefont
  {McAndrew}, \citenamefont {Bell},\ and\ \citenamefont
  {Garcia}}]{Mcandrew2007}%
  \BibitemOpen
  \bibfield  {author} {\bibinfo {author} {\bibnamefont {McAndrew},
  \bibfnamefont {F.~T.}}, \bibinfo {author} {\bibfnamefont {E.~K.}\
  \bibnamefont {Bell}}, \ and\ \bibinfo {author} {\bibfnamefont {C.~M.}\
  \bibnamefont {Garcia}}} (\bibinfo {year} {2007}),\ \href@noop {} {\bibfield
  {journal} {\bibinfo  {journal} {J Appl Soc Psychol}\ }\textbf {\bibinfo
  {volume} {37}}~(\bibinfo {number} {7}),\ \bibinfo {pages} {1562}}\BibitemShut
  {NoStop}%
\bibitem [{\citenamefont {Mcpherson}\ \emph {et~al.}(2001)\citenamefont
  {Mcpherson}, \citenamefont {Smith-Lovin},\ and\ \citenamefont
  {Cook}}]{McPherson2001}%
  \BibitemOpen
  \bibfield  {author} {\bibinfo {author} {\bibnamefont {Mcpherson},
  \bibfnamefont {M.}}, \bibinfo {author} {\bibfnamefont {L.}~\bibnamefont
  {Smith-Lovin}}, \ and\ \bibinfo {author} {\bibfnamefont {J.~M.}\ \bibnamefont
  {Cook}}} (\bibinfo {year} {2001}),\ \href@noop {} {\bibfield  {journal}
  {\bibinfo  {journal} {Annu Rev Sociol}\ }\textbf {\bibinfo {volume}
  {27}}~(\bibinfo {number} {1}),\ \bibinfo {pages} {415}}\BibitemShut {NoStop}%
\bibitem [{\citenamefont {Milinski}\ \emph {et~al.}(2001)\citenamefont
  {Milinski}, \citenamefont {Semmann}, \citenamefont {Bakker},\ and\
  \citenamefont {Krambeck}}]{Milinski2001}%
  \BibitemOpen
  \bibfield  {author} {\bibinfo {author} {\bibnamefont {Milinski},
  \bibfnamefont {M.}}, \bibinfo {author} {\bibfnamefont {D.}~\bibnamefont
  {Semmann}}, \bibinfo {author} {\bibfnamefont {T.~C.}\ \bibnamefont {Bakker}},
  \ and\ \bibinfo {author} {\bibfnamefont {H.~J.}\ \bibnamefont {Krambeck}}}
  (\bibinfo {year} {2001}),\ \href@noop {} {\bibfield  {journal} {\bibinfo
  {journal} {Proc Biol Sci}\ }\textbf {\bibinfo {volume} {268}}~(\bibinfo
  {number} {1484}),\ \bibinfo {pages} {2495}}\BibitemShut {NoStop}%
\bibitem [{\citenamefont {Nakamaru}\ and\ \citenamefont
  {Kawata}(2004)}]{Nakamaru2004}%
  \BibitemOpen
  \bibfield  {author} {\bibinfo {author} {\bibnamefont {Nakamaru},
  \bibfnamefont {M.}}, \ and\ \bibinfo {author} {\bibfnamefont
  {M.}~\bibnamefont {Kawata}}} (\bibinfo {year} {2004}),\ \href@noop {}
  {\bibfield  {journal} {\bibinfo  {journal} {Evol Ecol Res}\ }\textbf
  {\bibinfo {volume} {6}}~(\bibinfo {number} {2}),\ \bibinfo {pages}
  {261}}\BibitemShut {NoStop}%
\bibitem [{\citenamefont {Nowak}(2006)}]{Nowak2006a}%
  \BibitemOpen
  \bibfield  {author} {\bibinfo {author} {\bibnamefont {Nowak}, \bibfnamefont
  {M.~A.}}} (\bibinfo {year} {2006}),\ \href@noop {} {\bibfield  {journal}
  {\bibinfo  {journal} {Science (New York, N.Y.)}\ }\textbf {\bibinfo {volume}
  {314}}~(\bibinfo {number} {5805}),\ \bibinfo {pages} {1560}}\BibitemShut
  {NoStop}%
\bibitem [{\citenamefont {Nowak}\ and\ \citenamefont
  {Sigmund}(1998)}]{Nowak1998}%
  \BibitemOpen
  \bibfield  {author} {\bibinfo {author} {\bibnamefont {Nowak}, \bibfnamefont
  {M.~A.}}, \ and\ \bibinfo {author} {\bibfnamefont {K.}~\bibnamefont
  {Sigmund}}} (\bibinfo {year} {1998}),\ \href@noop {} {\bibfield  {journal}
  {\bibinfo  {journal} {Nature}\ }\textbf {\bibinfo {volume} {393}}~(\bibinfo
  {number} {6685}),\ \bibinfo {pages} {573}}\BibitemShut {NoStop}%
\bibitem [{\citenamefont {Nowak}\ and\ \citenamefont
  {Sigmund}(2005)}]{Nowak2005}%
  \BibitemOpen
  \bibfield  {author} {\bibinfo {author} {\bibnamefont {Nowak}, \bibfnamefont
  {M.~A.}}, \ and\ \bibinfo {author} {\bibfnamefont {K.}~\bibnamefont
  {Sigmund}}} (\bibinfo {year} {2005}),\ \href@noop {} {\bibfield  {journal}
  {\bibinfo  {journal} {Nature}\ }\textbf {\bibinfo {volume} {437}}~(\bibinfo
  {number} {7063}),\ \bibinfo {pages} {1291}}\BibitemShut {NoStop}%
\bibitem [{\citenamefont {Ohtsuki}\ and\ \citenamefont {Iwasa}(2004)}]{A2004}%
  \BibitemOpen
  \bibfield  {author} {\bibinfo {author} {\bibnamefont {Ohtsuki}, \bibfnamefont
  {H.}}, \ and\ \bibinfo {author} {\bibfnamefont {Y.}~\bibnamefont {Iwasa}}}
  (\bibinfo {year} {2004}),\ \href@noop {} {\bibfield  {journal} {\bibinfo
  {journal} {J Theor Biol}\ }\textbf {\bibinfo {volume} {231}}~(\bibinfo
  {number} {1}),\ \bibinfo {pages} {107}}\BibitemShut {NoStop}%
\bibitem [{\citenamefont {Ohtsuki}\ and\ \citenamefont {Iwasa}(2006)}]{A2006a}%
  \BibitemOpen
  \bibfield  {author} {\bibinfo {author} {\bibnamefont {Ohtsuki}, \bibfnamefont
  {H.}}, \ and\ \bibinfo {author} {\bibfnamefont {Y.}~\bibnamefont {Iwasa}}}
  (\bibinfo {year} {2006}),\ \href@noop {} {\bibfield  {journal} {\bibinfo
  {journal} {J Theor Biol}\ }\textbf {\bibinfo {volume} {239}}~(\bibinfo
  {number} {4}),\ \bibinfo {pages} {435}}\BibitemShut {NoStop}%
\bibitem [{\citenamefont {Paine}(1967)}]{Paine1967}%
  \BibitemOpen
  \bibfield  {author} {\bibinfo {author} {\bibnamefont {Paine}, \bibfnamefont
  {R.}}} (\bibinfo {year} {1967}),\ \href@noop {} {\bibfield  {journal}
  {\bibinfo  {journal} {Man}\ }\textbf {\bibinfo {volume} {2}}~(\bibinfo
  {number} {2}),\ \bibinfo {pages} {278}}\BibitemShut {NoStop}%
\bibitem [{\citenamefont {Piazza}\ and\ \citenamefont
  {Bering}(2008)}]{Piazza2008}%
  \BibitemOpen
  \bibfield  {author} {\bibinfo {author} {\bibnamefont {Piazza}, \bibfnamefont
  {J.}}, \ and\ \bibinfo {author} {\bibfnamefont {J.~M.}\ \bibnamefont
  {Bering}}} (\bibinfo {year} {2008}),\ \href@noop {} {\bibfield  {journal}
  {\bibinfo  {journal} {Evolution and Human Behavior}\ }\textbf {\bibinfo
  {volume} {29}}~(\bibinfo {number} {3}),\ \bibinfo {pages} {172}}\BibitemShut
  {NoStop}%
\bibitem [{\citenamefont {Radicchi}\ \emph {et~al.}(2007)\citenamefont
  {Radicchi}, \citenamefont {Vilone}, \citenamefont {Yoon},\ and\ \citenamefont
  {Meyer-Ortmanns}}]{Radicchi2007}%
  \BibitemOpen
  \bibfield  {author} {\bibinfo {author} {\bibnamefont {Radicchi},
  \bibfnamefont {F.}}, \bibinfo {author} {\bibfnamefont {D.}~\bibnamefont
  {Vilone}}, \bibinfo {author} {\bibfnamefont {S.}~\bibnamefont {Yoon}}, \ and\
  \bibinfo {author} {\bibfnamefont {H.}~\bibnamefont {Meyer-Ortmanns}}}
  (\bibinfo {year} {2007}),\ \href@noop {} {\bibfield  {journal} {\bibinfo
  {journal} {Phys Rev E}\ }\textbf {\bibinfo {volume} {75}}~(\bibinfo {number}
  {2}),\ \bibinfo {pages} {026106}}\BibitemShut {NoStop}%
\bibitem [{\citenamefont {Rugh}(1996)}]{Rugh1996}%
  \BibitemOpen
  \bibfield  {author} {\bibinfo {author} {\bibnamefont {Rugh}, \bibfnamefont
  {W.}}} (\bibinfo {year} {1996}),\ \href@noop {} {\emph {\bibinfo {title}
  {{Linear system theory}}}}\ (\bibinfo  {publisher} {Prentice Hall},\ \bibinfo
  {address} {Upper Saddle River})\BibitemShut {NoStop}%
\bibitem [{\citenamefont {Smith}\ and\ \citenamefont {{Maynard
  Smith}}(1982)}]{MaynardSmith1982}%
  \BibitemOpen
  \bibfield  {author} {\bibinfo {author} {\bibnamefont {Smith}, \bibfnamefont
  {J.~M.}}, \ and\ \bibinfo {author} {\bibfnamefont {J.}~\bibnamefont {{Maynard
  Smith}}}} (\bibinfo {year} {1982}),\ \href@noop {} {\emph {\bibinfo {title}
  {{Evolution and the Theory of Games}}}}\ (\bibinfo  {publisher} {Cambridge
  University Press},\ \bibinfo {address} {Cambridge})\BibitemShut {NoStop}%
\bibitem [{\citenamefont {Sommerfeld}\ \emph {et~al.}(2008)\citenamefont
  {Sommerfeld}, \citenamefont {Krambeck},\ and\ \citenamefont
  {Milinski}}]{Sommerfeld2009}%
  \BibitemOpen
  \bibfield  {author} {\bibinfo {author} {\bibnamefont {Sommerfeld},
  \bibfnamefont {R.~D.}}, \bibinfo {author} {\bibfnamefont {H.-J.}\
  \bibnamefont {Krambeck}}, \ and\ \bibinfo {author} {\bibfnamefont
  {M.}~\bibnamefont {Milinski}}} (\bibinfo {year} {2008}),\ \href@noop {}
  {\bibfield  {journal} {\bibinfo  {journal} {Proc Biol Sci}\ }\textbf
  {\bibinfo {volume} {275}}~(\bibinfo {number} {1650}),\ \bibinfo {pages}
  {2529}}\BibitemShut {NoStop}%
\bibitem [{\citenamefont {Sommerfeld}\ \emph {et~al.}(2007)\citenamefont
  {Sommerfeld}, \citenamefont {Krambeck}, \citenamefont {Semmann},\ and\
  \citenamefont {Milinski}}]{Semmann2007}%
  \BibitemOpen
  \bibfield  {author} {\bibinfo {author} {\bibnamefont {Sommerfeld},
  \bibfnamefont {R.~D.}}, \bibinfo {author} {\bibfnamefont {H.-J.}\
  \bibnamefont {Krambeck}}, \bibinfo {author} {\bibfnamefont {D.}~\bibnamefont
  {Semmann}}, \ and\ \bibinfo {author} {\bibfnamefont {M.}~\bibnamefont
  {Milinski}}} (\bibinfo {year} {2007}),\ \href@noop {} {\bibfield  {journal}
  {\bibinfo  {journal} {Proc Natl Acad Sci U S A}\ }\textbf {\bibinfo {volume}
  {104}}~(\bibinfo {number} {44}),\ \bibinfo {pages} {17435}}\BibitemShut
  {NoStop}%
\bibitem [{\citenamefont {Still}(2001)}]{Still2001}%
  \BibitemOpen
  \bibfield  {author} {\bibinfo {author} {\bibnamefont {Still}, \bibfnamefont
  {G.}}} (\bibinfo {year} {2001}),\ \href {\doibase 10.1080/02331930108844539}
  {\bibfield  {journal} {\bibinfo  {journal} {Optimization}\ }\textbf {\bibinfo
  {volume} {49}}~(\bibinfo {number} {4}),\ \bibinfo {pages} {387}}\BibitemShut
  {NoStop}%
\bibitem [{\citenamefont {Szell}\ \emph {et~al.}(2010)\citenamefont {Szell},
  \citenamefont {Lambiotte},\ and\ \citenamefont {Thurner}}]{Szell2010}%
  \BibitemOpen
  \bibfield  {author} {\bibinfo {author} {\bibnamefont {Szell}, \bibfnamefont
  {M.}}, \bibinfo {author} {\bibfnamefont {R.}~\bibnamefont {Lambiotte}}, \
  and\ \bibinfo {author} {\bibfnamefont {S.}~\bibnamefont {Thurner}}} (\bibinfo
  {year} {2010}),\ \href@noop {} {\bibfield  {journal} {\bibinfo  {journal}
  {Proc Natl Acad Sci U S A}\ }\textbf {\bibinfo {volume} {107}}~(\bibinfo
  {number} {31}),\ \bibinfo {pages} {13636}}\BibitemShut {NoStop}%
\bibitem [{\citenamefont {Traag}\ and\ \citenamefont
  {Bruggeman}(2009)}]{Traag2009}%
  \BibitemOpen
  \bibfield  {author} {\bibinfo {author} {\bibnamefont {Traag}, \bibfnamefont
  {V.~A.}}, \ and\ \bibinfo {author} {\bibfnamefont {J.}~\bibnamefont
  {Bruggeman}}} (\bibinfo {year} {2009}),\ \href@noop {} {\bibfield  {journal}
  {\bibinfo  {journal} {Phys Rev E}\ }\textbf {\bibinfo {volume}
  {80}}~(\bibinfo {number} {3}),\ \bibinfo {pages} {036115}},\ \Eprint
  {http://arxiv.org/abs/0811.2329} {arXiv:0811.2329} \BibitemShut {NoStop}%
\bibitem [{\citenamefont {Traag}\ \emph {et~al.}(2011)\citenamefont {Traag},
  \citenamefont {{Van Dooren}},\ and\ \citenamefont {Nesterov}}]{Traag2011a}%
  \BibitemOpen
  \bibfield  {author} {\bibinfo {author} {\bibnamefont {Traag}, \bibfnamefont
  {V.~A.}}, \bibinfo {author} {\bibfnamefont {P.}~\bibnamefont {{Van Dooren}}},
  \ and\ \bibinfo {author} {\bibfnamefont {Y.}~\bibnamefont {Nesterov}}}
  (\bibinfo {year} {2011}),\ in\ \href@noop {} {\emph {\bibinfo {booktitle}
  {IEEE Symposium on Artificial Life 2011}}}\ (\bibinfo  {publisher} {IEEE},\
  \bibinfo {address} {Piscataway})\ pp.\ \bibinfo {pages}
  {154--161}\BibitemShut {NoStop}%
\bibitem [{\citenamefont {Wert}\ and\ \citenamefont {Salovey}(2004)}]{Wert}%
  \BibitemOpen
  \bibfield  {author} {\bibinfo {author} {\bibnamefont {Wert}, \bibfnamefont
  {S.~R.}}, \ and\ \bibinfo {author} {\bibfnamefont {P.}~\bibnamefont
  {Salovey}}} (\bibinfo {year} {2004}),\ \href@noop {} {\bibfield  {journal}
  {\bibinfo  {journal} {Rev Gen Psychol}\ }\textbf {\bibinfo {volume}
  {8}}~(\bibinfo {number} {2}),\ \bibinfo {pages} {122}}\BibitemShut {NoStop}%
\bibitem [{\citenamefont {Wilson}(1975)}]{Wilson1975}%
  \BibitemOpen
  \bibfield  {author} {\bibinfo {author} {\bibnamefont {Wilson}, \bibfnamefont
  {D.~S.}}} (\bibinfo {year} {1975}),\ \href@noop {} {\bibfield  {journal}
  {\bibinfo  {journal} {Proc Natl Acad Sci U S A}\ }\textbf {\bibinfo {volume}
  {72}}~(\bibinfo {number} {1}),\ \bibinfo {pages} {143}}\BibitemShut {NoStop}%
\end{thebibliography}%

\section*{Materials and Methods}

% Already renew figures
\renewcommand{\thefigure}{S\arabic{figure}}    
\setcounter{figure}{0}

\begin{figure}[t] 
  \begin{center} 
    \includegraphics{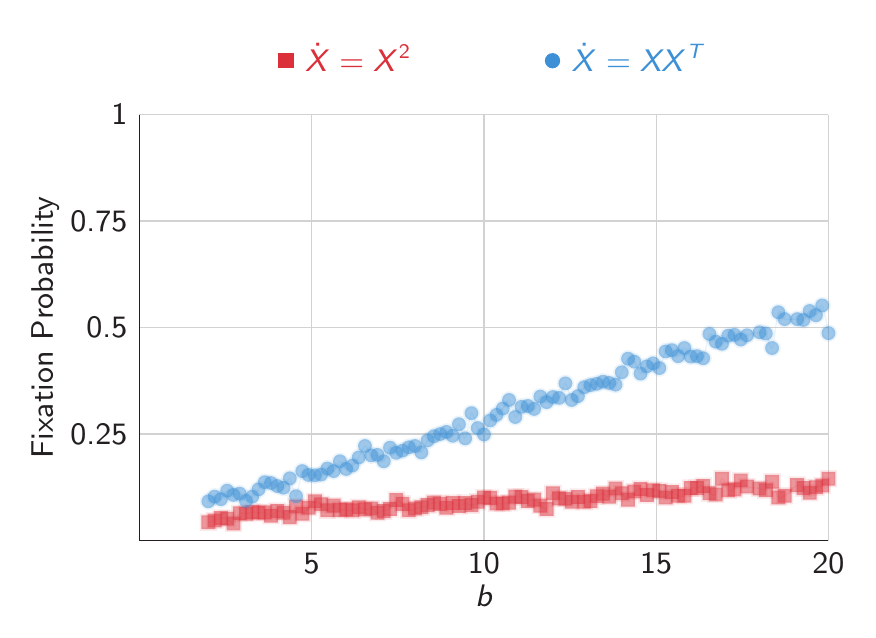} 
  \caption{Results including type A, B and defectors.\label{fig:defect_ABD}}
  \end{center}
\end{figure}

\begin{figure*}[t]
  \begin{center}
    \includegraphics[width=\textwidth]{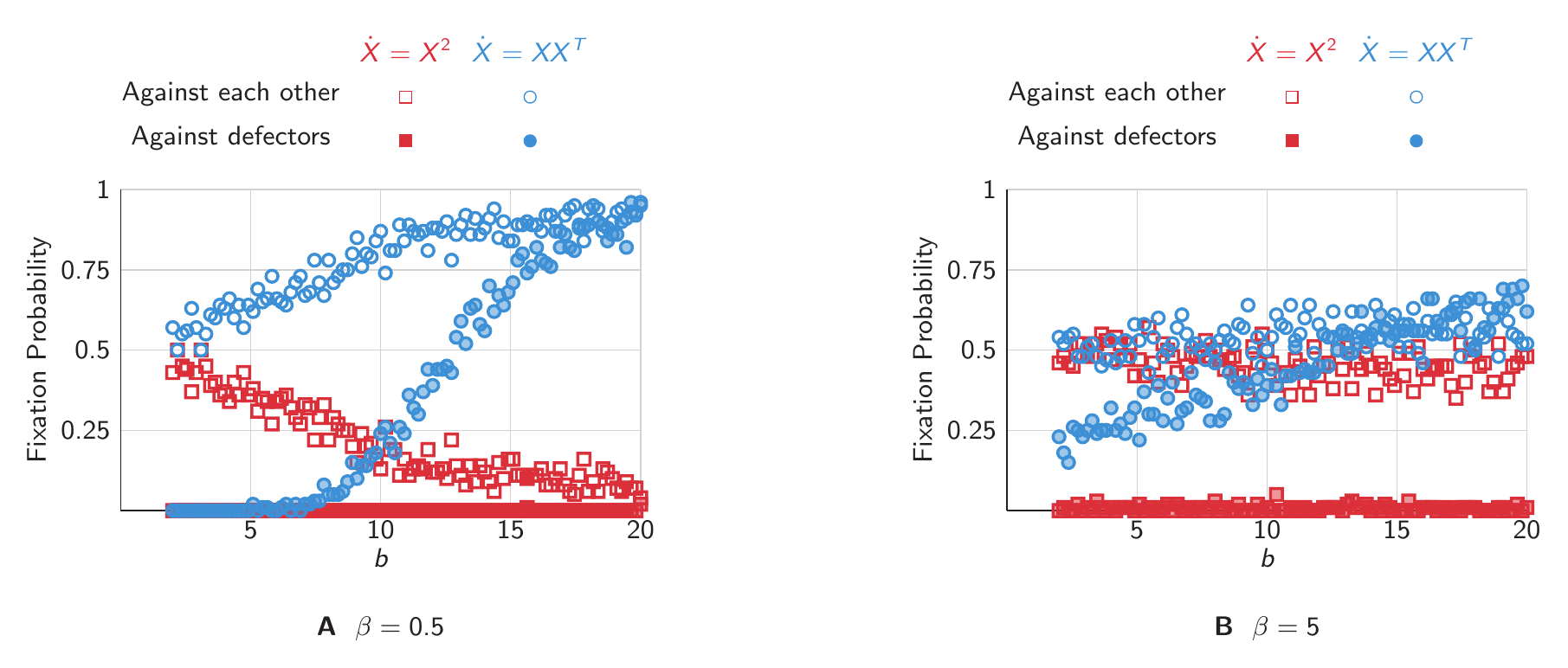}
  \caption{Results different intensities of selection.\label{fig:beta}}
  \end{center}
\end{figure*}

In the simulations of the evolution of cooperation, the dynamics consist of
two parts: (1) the interaction dynamics within each generation; and (2) the
dynamics prescribing how the population evolves from generation to generation.

\subsection{Interaction Dynamics} 
We include three possible types of agents in our simulations:
\begin{description}[font=\sffamily\small]
  \item[Type A] uses $\dot{X} = X^2$ dynamics,
  \item[Type B] uses $\dot{X} = XX^T$ dynamics, and
  \item[Defectors] have trivial reputation dynamics ${\dot X}=0$, with
    negative constant reputations.
\end{description}
We can decompose the reputation matrix $X(t)$ accordingly into three parts: 
\begin{equation*} 
  X(t) = \left(\begin{matrix} X_A(t) \\ X_B(t) \\ X_D(t) \end{matrix}\right), 
\end{equation*}
where $X_A(t)$ are the reputations of all agents in the eyes of agents of type
A, $X_B(t)$ for type B and $X_D(t)$ for defectors. The reputations $X_A(0)$ and
$X_B(0)$ are initialized from a standard Gaussian distribution. The initial
reputation for $X_D(0)$ will be set to a fixed negative value. To be clear,
$X_D(0)$ is the reputation of all other agents in the eyes of defectors, which
is negative initially. The initial reputation of the defectors themselves
is of course not necessarily negative initially. For the results displayed here
we have used $X_D(0) = -10$, but results remain by and large the same when
varying this parameter, as long as it remains sufficiently negative. 

Since we are dealing with continuous dynamics in this paper, we assume all
agents are involved in infinitesimally short games at each time instance $t$.
Each agent $i$ may choose to either cooperate or defect with another agent $j$,
and this decision may vary from one agent to the next. For agents of type A and
type B the decision to cooperate is based on the reputation: they defect
whenever $X_{ij}(t) \leq 0$ and cooperate whenever $X_{ij}(t) > 0$. We define
the cooperation matrix $C(t)$ accordingly 
\begin{equation*} 
  C_{ij}(t) =
  \begin{cases} 
    0 & \text{if~} X_{ij} \leq 0 \\ 
    1 & \text{if~} X_{ij} > 0 \\
  \end{cases} 
\end{equation*} 
Defectors will simply always defect. Whenever an agent $i$ cooperates with $j$
the latter receives a payoff of $b$ at a cost of $c$ to agent $i$. We integrate
this payoff over all infinitesimally short games from time $0$ to time $t^*$,
which can be represented as
\begin{equation*}
  P(g) = \frac{1}{n}\int_0^{t^*} b C(t)^T e - c C(t) e dt,
\end{equation*}
where $e = (1,\ldots,1)$ the vector of all ones for a certain generation $g$. 

\subsection{Evolutionary Dynamics} 
We have simulated the evolution of cooperation for $n=10,20,\ldots,60$ agents, which stays
constant throughout evolution. We consider four different schemes for
initializing the first generation:
\vskip1em
\begin{tabular}{lllll}
     &                        & $p_A(0)$ & $p_B(0)$ & $p_D(0)$ \\
  1) &       Type A vs Type B &    $1/2$ &    $1/2$ &        - \\
  2) &    Type A vs Defectors &    $1/2$ &        - &    $1/2$ \\
  3) &    Type B vs Defectors &        - &    $1/2$ &    $1/2$ \\
  4) & Type A,B and Defectors &    $1/3$ &    $1/3$ &    $1/3$ 
\end{tabular}
\vskip1em
\noindent Here $p_A(0)$,$p_B(0)$ and $p_D(0)$ are respectively the proportion of
agents of type A, type B and defectors in the first generation. We use the
vector $T_i(g) \in \{A,B,D\}$ to denote the type of agent $i$ in generation $g$,
so that $T_i(g) = A$ if agent $i$ is a type A player, $T_i(g) = B$ for a type B
player, and $T_i(g) = D$ for a defector. We are interested in estimating the
probability that a single type takes over the whole population, known as the
fixation probability $\rho_A$, $\rho_B$ and $\rho_D$ for the three different
types. If a type has no evolutionary advantage, it is said to be evolutionary
neutral, and in that case its fixation probability is equal to its initial
frequency, e.g. for type A $\rho_A = p_A(0)$.

We will keep the population constant at the initial $n$, and simply choose $n$ new
agents according to their payoff for the next generation. This can be thought of
as choosing $n$ times one of the $n$ agents in the old generation for
reproduction. Let $\phi_i$ denote the probability that an agent is selected for
reproduction, which we define as
\begin{equation*}
  \phi_i = \frac{e^{\beta P_i(g)}}{\sum_i e^{\beta P_i(g)}}. 
\end{equation*}
Since we are only interested in the number of agents of a certain type, we can
also gather all payoffs for the same type of agents, and write
\begin{equation*}
  \Phi_q = \sum_{i : T_i(g) = q} \phi_i,
\end{equation*}
where $q \in \{A,B,D\}$ represents the type of agent. The probability to select a
type A agent, a type B agent or a defector is then respectively $\Phi_A$,
$\Phi_B$ and $\Phi_D$.  In the next generation, the probability that agent $i$
is of a specific type $q$ can then be written as 
\begin{equation*}
  \Pr(T_i(g+1) = q) = \Phi_q.
\end{equation*}
This evolutionary mechanism can be seen as a Wright-Fisher
process~\cite{Gillespie2004} with fitnesses $e^{\beta P_i(g)}$. It is well known
that this process converges faster than a Moran birth-death process, since it
essentially takes $n$ time steps in a Moran process to reproduce the effect of
one time step in a Wright-Fisher process~\cite{Gillespie2004}.  Because of the
high computational costs (solving repeatedly a non-linear system of differential
equations of size $n^2$), this process is preferable.

Higher $\beta$ signifies higher selective pressure, and leads to a higher
reproduction of those with a high payoff, and in the case that $\beta \to
\infty$ only those with the maximum payoff reproduce. On the other hand, for
$\beta \to 0$ this tends to the uniform distribution $\phi_i = 1/n$, where
payoffs no longer play any role. We have used $\beta=0.5$ for the low selective
pressure, $\beta=5$ for the high selective pressure, reported in the SI. For the
results in the main text we have used $\beta=1$.

For an evolutionary neutral selection in where all $P_i(g) = P$ are effectively
the same, $\beta$ has no effect, and $\phi_i = 1/n$. Notice that if we rescale
$P_i(g)$ by $1/c$ so that the payoff effectively becomes
\begin{equation*}
  \frac{1}{c} P_i(g) = \frac{1}{n}\int_0^{t^*} \frac{b}{c} C(t)^T e - C(t) e dt,
\end{equation*}
and we rescale $\beta$ by $c$, then the reproduction probabilities remain
unchanged. Hence, only the ratio $b/c$ effectively plays a role up to a
rescaling of the intensity of selection. Since the point at which the evolution
is neutral (i.e. $\rho$ equals the initial proportional frequency), is
independent of $\beta$, this point will only depend on the ratio $b/c$. So, we
normalized the cost $c=1$. To verify this, we also ran additional simulations
with different costs, which indeed gave the same results. 

We stop the simulation whenever one of the types becomes fixated in the
population. With \emph{fixation} we mean that all other types have gone extinct,
and only a single type remains. If no type has become fixated after
\numprint{1,000} generations, we terminate the simulation and count as winner
the most frequent type. This almost never happens, and the simulation usually
stops after a relatively small number of generations. 

In total, we repeat this process \numprint{1,000} times for the results in the
main text, and for the low ($\beta=0.5$) and high ($\beta=5$) selective pressure
$100$ times. This means that we run the evolutionary dynamics until one of the
types has become fixated, and we record which type has ``won''.  After that, we
again start from the first generation, and run until fixation, and repeat this.
Finally, we calculate how many rounds a type has ``won'' compared to the total
number of rounds, which yields the fixation probability $\rho$.

\onecolumngrid
\clearpage
\twocolumngrid
\renewcommand{\theequation}{S\arabic{equation}}    
% redefine the command that creates the equation no.    
\setcounter{equation}{0}  % reset counter
\setcounter{section}{0}

\onecolumngrid
  \begin{center}
    \Huge\bfseries\sffamily
    Supplementary Information
  \end{center}
  \vskip10ex
\twocolumngrid

\section{Preliminaries}
We investigate matrix differential equations of the form ${\dot X}=F(X,X^T)$,
where $X$ is a real $n\times n$ matrix, and $F$ is a one of two specific, smooth
functions. These functions are such that it turns out to be advantageous to
consider the dynamics of the symmetric and skew-symmetric parts of $X$. Recall
that $\reals^{n\times n}={\cal S} \oplus {\cal A}$, where ${\cal S}$ is the
linear subspace of real symmetric matrices, and ${\cal A}$ the linear subspace
of skew-symmetric matrices. Thus, given any $X\in \reals^{n\times n}$, we can
find unique symmetric $S\in {\cal S}$ and skew-symmetric $A\in {\cal A}$ such
that $X=S+A$. More explicitly, $S=(X+X^T)/2$ and $A=(X-X^T)/2$. Moreover, using the
inner product $\langle X,Y \rangle=\trace(XY^T)$, there holds that
\begin{equation}\label{perp}
{\cal A}^\perp ={\cal S}.
\end{equation}
The norm induced by this inner product is the Frobenius norm
$|X|_F=(\trace(XX^T))^{\frac{1}{2}}$. Recall that the Frobenius norm is
unitarily invariant, i.e. if $U$ is orthogonal (i.e. $UU^T=I_n$), then
\begin{equation}\label{norm-invar}
|UXU^T|_F=|X|_F.
\end{equation}
We denote by $I_n$ the $n\times n$ identity matrix, and by $J_n$ a specific skew
symmetric matrix:
\begin{equation}\label{skew_J}
  J_{n} = \begin{pmatrix} 
      0 & I_{n/2} \\
      -I_{n/2} & 0 \\
    \end{pmatrix}, n \text{~even}.
\end{equation}
For all other linear algebra related terminology and properties we refer to
\cite{Horn1985}.

We briefly review two key ingredients of Heider's (static) theory on social
balance, namely those of a {\it balanced triangle} and a {\it balanced
network}:
\begin{definition}\label{balance}
  A triangle of (not necessarily distinct) agents $i,j$ and $k$ is called
  balanced if
  \begin{equation}\label{rel3}
    X_{ij}X_{ik}X_{kj}>0.
  \end{equation}
  A network is said to be balanced if all triangles of agents in the network are
  balanced.  
\end{definition} 
It turns out that a balanced network takes on a specific
structure, in that at most 2 factions emerge, where members within each faction
have positive opinions about each other, but members in different factions have
negative opinions about each other. This result is known as the Structure
Theorem~\cite{Harary1953,Cartwrighta}:
\begin{theorem}[Structure Theorem in \cite{Harary1953,Cartwrighta}]
  \label{factions}  Let $X$ represent a balanced network. Then up to a
  permutation of agents, the matrix $X$ has the following sign structure:
  $$
  \begin{pmatrix}
  +
  \end{pmatrix}\textrm{ or }
  \begin{pmatrix}
  +&-\\
  -&+
  \end{pmatrix}.
  $$
  Conversely, if, up to permutation, $X$ has one of these structures, then it represents a balanced network.
\end{theorem}
Notice that the same theorem holds irrespective of any permutation of $i$,$j$
and $k$ in definition~\ref{balance}.

\section{Equation ${\dot{X}=X^2}$}
Consider the model studied numerically in \cite{Kulakowski2005} and analysed for
symmetric initial conditions in \cite{Marvel2011}:
\begin{equation}\label{flow-X^2}
{\dot X}=X^2, X(0)=X_0,
\end{equation}
where each $X_{ij}$ is real-valued and denotes the opinion agent $i$ has about
agent $j$. Positive values mean that agent $i$ thinks favourably about $j$,
whereas negative values mean that $i$ thinks unfavourably about $j$. More
explicitly, model~\ref{flow-X^2} can also be written entrywise:
\begin{equation}\label{entry-wise-X^2}
{\dot X}_{ij}=\sum_k X_{ik}X_{kj}.
\end{equation}
The basic question in this context is whether or not the solutions
of~\ref{flow-X^2} evolve towards a state which corresponds to a balanced
network. A minor technical issue is that the solution $X(t)$ of
\ref{flow-X^2} often blows up in finite time ${\bar t}$ as we shall see
later. To resolve this problem we investigate the sign pattern of the matrix
limit $\lim_{t\rightarrow {\bar t}}X(t)/|X(t)|_F$ instead, and say that the
network evolves to a balanced state, if this matrix limit is balanced.

\subsection*{Normal initial condition}
We start by defining
$$
{\cal N}=\{X\in \reals^{n\times n}| XX^T=X^TX\},
$$
the set of real, normal matrices. Notice that if $X$ belongs to $\cal{N}$ then
so does $X^2$, hence the set $\cal{N}$ is invariant for $\dot{X}=X^2$.

Recall that normal matrices are (block)-diagonalizable with blocks of size at
most $2$ by an orthogonal transformation: if $X_0\in {\cal N}$, then
\begin{equation}
  U^TX_0U=\Lambda_0,
  \label{diagonalizable}
\end{equation} 
where $\Lambda_0$ consists of real $1 \times 1$ scalar
blocks $A_i$ and real $2\times 2$ blocks $B_j=\alpha_j I_2 + \beta_j J_2$ with
$\beta_j\neq 0$.

Note that if $\Lambda(t)$ is the solution to the initial value problem ${\dot
\Lambda}=\Lambda^2$, $\Lambda(0)=\Lambda_0$, then $X(t):=U\Lambda(t)U^T$ is the
solution to Eq.~\ref{flow-X^2}. This shows it is sufficient to solve system
\ref{flow-X^2} in case of scalar $X$ or in case of a specific, $2\times 2$,
normal matrix $X$. The scalar case is easy to solve: the solution of ${\dot
x}=x^2$, $x(0)=x_0$, is
\begin{equation}
  x(t)= \frac{x_0}{1-x_0t},
    \label{equ:scalar}
\end{equation}
which is easily verified, so we turn to the $2\times 2$
case by considering: 
\begin{equation}\label{flow2-X^2}
{\dot X}=X^2,\; X(0)=\alpha I_2 + \beta J_2,
\textrm{ where } \beta > 0.
\end{equation}

\begin{lemma}\label{2by2-X^2}
The forward solution $X(t)$ of~\ref{flow2-X^2} is defined for all $t\in [0,+\infty)$, and 
$$
\lim_{t\rightarrow +\infty} X(t)=0
\textrm{ and }\lim_{t\rightarrow +\infty} \frac{X(t)}{|X(t)|_F}=-\frac{\sqrt{2}}{2}I_2.
$$
\end{lemma}
\begin{proof}
Let $X_0=S_0+A_0$, $S_0=\alpha I_2$ and $A_0=\beta J_2$ where $J_2$ is as
defined in Eq.~\ref{skew_J}. Then the solution $X(t)$ of~\ref{flow2-X^2}
can be decomposed as $S(t)+A(t)$, where
\begin{align}
{\dot S}&=S^2+A^2,\;\; S(0)=S_0,\label{s-part-X^2}\\
{\dot A}&=AS+SA,\;\; A(0)=A_0\label{a-part-X^2}.
\end{align}
Note that~\ref{s-part-X^2} is a matrix Riccati differential equation with
the property that the set ${\cal L}:=\{s I_2+a J_2| s,\; a \in \reals\}$, is an
invariant set under the flow. Therefore it suffices to solve the scalar Riccati
differential equation corresponding to the dynamics of the scalar coefficients
$s$ and $a$:
\begin{align}
{\dot s}&=s^2-a^2,\; s(0)=\alpha, \label{s-sys}\\
{\dot a}&=2as,\; a(0)=\beta, \label{a-sys}
\end{align}
whose solution is given implicitly by:
$$ s^2+\left(a-\frac{1}{2c}\right)^2=\left(\frac{1}{2c}\right)^2\textrm{ if }
c\neq 0, $$
where $c$ is an integration constant. So, the orbits form circles which are
centred at $(0,1/2c)$ and pass through $(0,0)$, and by $a=0$ if $c=0$. The phase
portrait of system~\ref{s-sys}-\ref{a-sys} is illustrated in Fig.~\ref{phase}. 

\begin{figure}[t]
\begin{center}
  \includegraphics{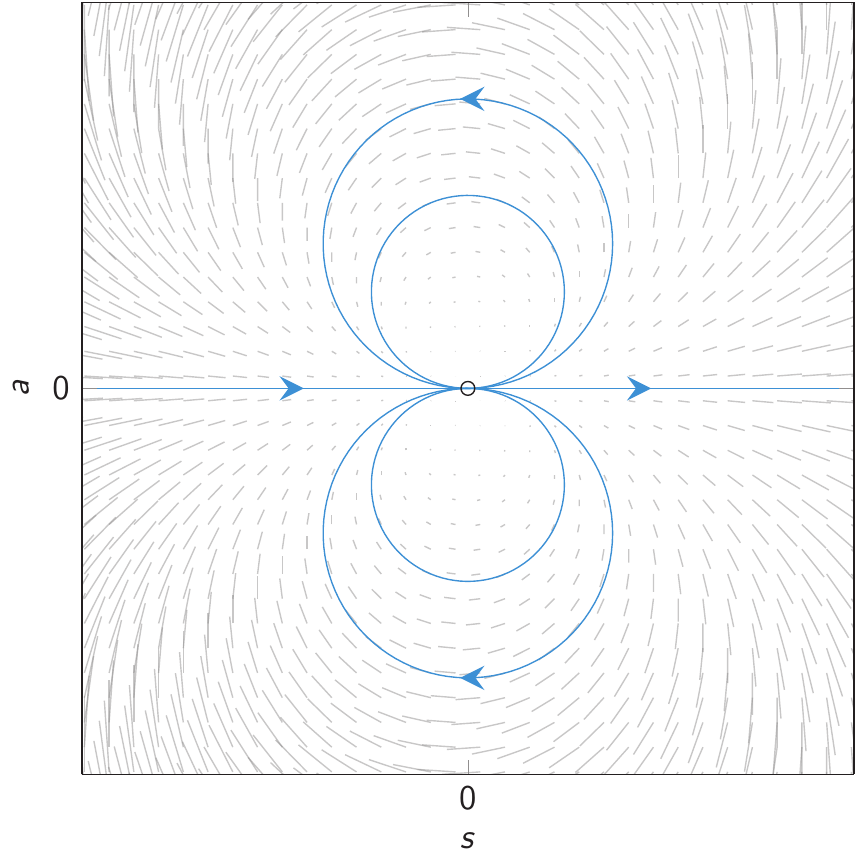}
  \caption{Phase portrait of system~\ref{s-sys}-\ref{a-sys}. Circular orbits
  in the upper half plane ($a>0$) are traversed counter clockwise, whereas
  circular orbits in the lower half plane ($a<0$) are traversed clockwise.\label{phase}} 
\end{center}
\end{figure}

All solutions $(s(t),a(t))$ of system~\ref{s-sys}-\ref{a-sys}, not
starting on the $s$-axis, converge to zero as $t\rightarrow +\infty$, and
approach the origin in the second quadrant for solutions in the
upper-half-plane, and in the third quadrant for solutions in the
lower-half-plane. Moreover, since the $s$-axis is the tangent line to every
circular orbit at the origin, the slopes $a(t)/s(t)$ converge to $0$ along every
solution $\lim_{t\rightarrow +\infty}a(t)/s(t)=0$.
Consequently, the forward solution $X(t)$ of~\ref{flow2-X^2} satisfies:
$$\lim_{t\rightarrow +\infty}X(t)
  =\lim_{t\rightarrow +\infty}s(t)I_2 + a(t)J_2
  =0,$$
and 
$$\lim_{t\rightarrow +\infty}\frac{X(t)}{|X(t)|_F}=-\frac{\sqrt{2}}{2}I_2.$$
\end{proof}

Combining the solution for the scalar and $2 \times 2$ case yields our main
result in the normal case:
\begin{theorem}\label{main-X^2}
Let $X_0\in {\cal N}$, and let $(U,\Lambda_0)$ be as in Eq.
$\ref{diagonalizable}$. Define
$$
{\bar t}_i=\begin{cases}
1/a_i\textrm{ if } a_i>0\\
+\infty \textrm{ if } a_i\leq 0
\end{cases}
\textrm{ for all } i=1,\dots,k, 
$$
and let $ {\bar t}=\min_i {\bar t}_i$. Then the forward solution $X(t)$ of~\ref{flow-X^2} is defined for $[0,{\bar
t})$. 

If there is a unique $i^*\in \{1,\dots,k\}$ such that ${\bar t}={\bar t}_{i^*}$
is finite, then
$$
\lim_{t\rightarrow {\bar t}_{i^*}-}\frac{X(t)}{|X(t)|_F}=U_{i^*}U_{i^*}^T,
$$
where $U_{i^*}$ is the $i^*$th column of $U$, an eigenvector corresponding to eigenvalue $a_{i^*}$ of $X_0$. 
\end{theorem}
\begin{proof}
Consider the initial value problem:
$$
{\dot \Lambda}=\Lambda^2,\; \Lambda(0)=\Lambda_0.
$$
Its solution is given by
$$
\Lambda(t)=\begin{pmatrix}
\frac{a_1}{1-a_1t}&\dots&0 &0&\dots& 0\\
\vdots&  \ddots &\vdots& \vdots&\ddots&\vdots\\
0&\dots& \frac{a_k}{1-a_kt}&0&\dots&0\\
0&\dots&0& X_1(t)&\dots&0\\
\vdots&\ddots&\vdots&\vdots&\ddots&\vdots\\
0&\dots&0&0&\dots&X_l(t)
\end{pmatrix},
$$
where for all $j=1,\dots,l$, $X_j(t)$ is the forward solution
of~\ref{flow2-X^2}, which is defined for all $t$ in $[0,+\infty)$, and
converges to $0$ as $t\rightarrow +\infty$ by Lemma $\ref{2by2-X^2}$.

This clearly shows that $\Lambda (t)$ is defined in forward time for $t$ in
$[0,{\bar t})$. Since the solution of~\ref{flow-X^2} is given by
$X(t)=U\Lambda(t)U^T$, $X(t)$ is also defined in forward time for $t$ in
$[0,{\bar t})$. It follows from
~\ref{norm-invar} that
$$
\frac{X(t)}{|X(t)|_F}=U\frac{\Lambda(t)}{|\Lambda (t)|_F}U^T.
$$
If $i^*\in\{1,\dots,k\}$ is the unique value such that ${\bar t}={\bar t}_{i^*}$, then using~\ref{norm-invar}:
$$
\lim_{t\rightarrow {\bar t}_i^*}\frac{X(t)}{|X(t)|_F}=U\lim_{t\rightarrow {\bar t}_i^*}\frac{\Lambda(t)}{|\Lambda (t)|_F}U^T=
Ue_{i^*}e_{i^*}^TU^T=U_{i^*}U_{i^*}^T,
$$
where $e_{i^*}$ denotes the $i^*$th standard unit basis vector of $\reals^n$.
\end{proof}

Theorem $\ref{main-X^2}$ provides a sufficient condition guaranteeing that
social balance in the sense of definition $\ref{balance}$ is achieved. If
$X_0$ has a simple, positive, real eigenvalue $a_{i^*}$, and if no entry of the
eigenvector $U_{i^*}$ is zero, then the network becomes  balanced. Indeed, there
holds that, up to a permutation of its entries, the sign pattern of the
eigenvector $U_{i^*}$ is either: 
$$ U_{i^*}=\begin{pmatrix}+\end{pmatrix}\textrm{ or
  }\begin{pmatrix}-\end{pmatrix} \;\; \Longrightarrow \;\;
    U_{i^*}U_{i^*}^T=\begin{pmatrix}+\end{pmatrix}, $$
or
$$
U_{i^*}=\begin{pmatrix}
+\\
\hline
-
\end{pmatrix}\;\; \Longrightarrow \;\;U_{i^*}U_{i^*}^T=\begin{pmatrix}
+&-\\
-&+
\end{pmatrix}.
$$
In either case, Theorem $\ref{factions}$ implies that the normalized state of the system becomes balanced in finite time.

\subsection*{Generic initial condition}
Although Theorem $\ref{main-X^2}$ provides a sufficient condition for the
emergence of social balance, it requires that the initial condition $X_0$ is
normal. But the set ${\cal N}$ of normal matrices has measure zero in the set of
all real $n\times n$ matrices, and thus the question arises if social balance
will arise for non-normal initial conditions as well. We investigate this issue
here, and will see that generically, social balance is not achieved.

If $X_0$ is a general real $n\times n$ matrix, we can put it in real Jordan
canonical form by means of a similarity transformation: 
\begin{equation}\label{jor}
X(0)=T\Lambda_0T^{-1},\;\; TT^{-1}=I_n, 
\end{equation}
with $\Lambda_0 = \diag(A_1,\ldots, A_k, B_1,\ldots,B_l)$, where $A_i$ are real
Jordan blocks and
\begin{equation}\label{block-jor}
B_j=\begin{pmatrix}
C_i&I_2&\dots&0\\
0&C_i&\ddots&0\\
\vdots&\vdots& \ddots&\vdots\\
0&0&\dots&C_i
\end{pmatrix},
C_j=\alpha_j I_2 + \beta_j J_2,
\end{equation}
with $\beta_j\neq 0$.

We again observe that if $\Lambda(t)$ is the solution to the initial value
problem ${\dot \Lambda}=\Lambda^2$, $\Lambda(0)=\Lambda_0$, then
$X(t):=T\Lambda(t)T^{-1},$ is the solution to Eq.~\ref{flow-X^2}. Again, it is
sufficient to solve system~\ref{flow-X^2} in case of specific
block-triangular $X$ of the form $A_i$ or $B_j$ as in~\ref{block-jor}.  To
deal with the first form $A_i$, we first we consider more general, triangular
Toeplitz initial conditions: 
\begin{equation}\label{toep}
X(0)=\begin{pmatrix}\
x_1(0) & x_2(0) & \cdots & x_n(0)\\
    0  & x_1(0) & \ddots & x_{n-1}(0)\\
\vdots & \vdots & \ddots & \vdots \\
    0  &     0  & \cdots & x_1(0)
\end{pmatrix},
\end{equation}
with $x_i(0)$ reals, and denote ${\cal T}{\cal T}=\{X\;|\;X\textrm{ is of the
form }\ref{toep}\}$.  It turns out that this is an invariant set for the
system, which can be easily verified by noting that if $X$ belongs to ${\cal
T}{\cal T}$, then so does $X^2$.
\begin{lemma}\label{sol-real-Jordan}
Let $X(0)\in {\cal T}{\cal T}$ with 
$$
x_i(0)=\begin{cases}
a\neq0\textrm{ if } i=1\\
1\textrm{ if }i=2\\
0\textrm{ otherwise} 
\end{cases}.
$$
Then the forward solution $X(t)$ of~\ref{flow-X^2} is defined on $[0,t^*)$
where $t^*=1/a$ if $a > 0$ and on $t^* = \infty$ if $a \leq 0$, belongs to
${\cal T}{\cal T}$, and satisfies $$
x_i(t)=p_i\left( \frac{1}{1-at}\right),\;\; t\in [0,t^*),
$$
where each $p_i(z)$ is a polynomial of degree $i$:
\begin{equation}\label{poly2}
p_i(z)=\begin{cases}
az\textrm{ if }i=1\\
\frac{1}{a^{i-2}}z^i+\dots + c_iz^2\textrm{ otherwise }
\end{cases},
\end{equation}
where $c_i$ is some real constant, 
so that $p_i(z)$ has no constant or first order terms when $i>1$.
\end{lemma} 
\begin{proof}
First note that system~\ref{flow-X^2} can be solved recursively, 
starting with $x_1(t)$, followed by $x_2(t),x_3(t),\dots$. Only the first
equation for $x_1$ is nonlinear, whereas the equations for $x_2,x_3,\dots$ are
linear. To see this, we write these equations:
\begin{equation*}\label{recursion}
  {\dot x_i=\begin{cases}x_1^2,\;\textrm{ if }i=1\\
    (2x_1(t))x_2,\;\textrm{ if }i=2\\
    (2x_1(t))x_i+\sum_{k=2}^{i-1}x_k(t)x_{i-(k-1)}(t),\;\textrm{ if }i>2
  \end{cases}},
\end{equation*}
with $x_1(0) = a$, $x_2(0) = 1$ and $x_i(0) = 0$ for $i>2$.
The forward solution for $x_1$ is $x_1(t)=\frac{a}{1-at}$, for $t\in [0,t^*)$,
which establishes the result if $i=1$. The forward solution for $x_2$ is:
$x_2(t)=\frac{1}{(1-at)^2}$, for $t\in [0,t^*)$, which establishes the result if
$i=2$. If $i>2$, we obtain the proof by induction on $n$. Assume the result
holds for $i=1,\dots,n$, for some $n\geq 2$, and consider the equation for
$x_{n+1}$. 
Using that $x_n(0)=0$ for $n\geq 2$, the solution is given by:
\begin{multline*}
x_{n+1}(t)=\exponent^{\int_0^t2x_1(s)ds}\Biggl[0 + \\
    \int_0^t \left(\sum_{k=2}^{n}x_k(s)x_{n-k+2}(s)\right)
      \exponent^{\int_0^s-2x_1(\tau)d\tau}ds \Biggr].
\end{multline*}
Since $\exponent^{\int_0^t2x_1(s)ds}=x_2(t)$ and thus
$\exponent^{\int_0^s-2x_1(\tau)d\tau}=1/x_2(s)$, it follows that:
\begin{multline*}
x_{n+1}(t)=\frac{1}{(1-at)^2}\Biggl[\int_0^t \Bigl(\sum_{k=2}^{n}p_k(1/(1-as))\\
        p_{n-k+2}(1/(1-as))\Bigr) (1-as)^2 ds\Biggr].
\end{multline*}
Since the polynomials  appearing in the integral take the form of Eq.~$\ref{poly2}$, they are all missing first order and constant terms, and thus there follows that
\begin{multline*}
x_{n+1}(t)=\frac{1}{(1-at)^2}\Biggl[  \int_0^t \Bigl(\sum_{k=2}^{n} 
  \frac{1}{a^{n-2}}\frac{1}{(1-as)^{n+2}}+\\
\dots +c_kc_{n-k+2}\frac{1}{(1-as)^4} \Bigr)(1-as)^2 ds\Biggr]
\end{multline*}
and so that
$$x_{n+1}(t)=\frac{1}{a^{n-1}}\frac{1}{(1-at)^{n+1}}+\dots+\frac{c_{n+1}}{(1-at)^2},
\;\; t\in[0,t^*),$$
where $K_{n+1}$ and $c_{n+1}$ are certain constants (which are related in some
way which is irrelevant for what follows). This shows that $x_{n+1}(t)$ is
indeed of the form $p_{n+1}(1/(1-at))$ with $p_{n+1}(z)$ as in~\ref{poly2}. 
\end{proof}

Next we consider equation~\ref{flow-X^2} in case $X(0)$ is 
a block triangular Toeplitz initial condition:
\begin{equation}\label{block-toep}
X(0)=\begin{pmatrix}\
B_1(0) & B_2(0) & \cdots & B_n(0)\\
     0 & B_1(0) & \ddots & B_{n-1}(0)\\
\vdots & \vdots & \ddots & \vdots \\
     0 &     0  & \cdots & B_1(0)
\end{pmatrix},
\end{equation}
with $B_i(0) = \alpha_i I_2 + \beta_i J_2$ with $\alpha_i,\beta_i \in \reals$,
and denote ${\cal B}{\cal T}{\cal T}=\{X\;|\;X\textrm{ is of the form
}\ref{block-toep}\}$. Again the set ${\cal BTT}$ is invariant for system
~\ref{flow-X^2}. We use this to solve equation~\ref{flow-X^2} in case
$X(0)$ is a real Jordan block corresponding to a pair of eigenvalues $\alpha\pm
j\beta$.
\begin{lemma}\label{sol-complex-Jordan}
Let $X(0)\in {\cal BTT}$ with 
$$
B_i(0)=\begin{cases}
\alpha I_2 + \beta J_2 \textrm{ if } i=1\\
I_2\textrm{ if }i=2\\
0\textrm{ otherwise} 
\end{cases}.
$$
Then the forward solution $X(t)$ of~\ref{flow-X^2} is defined on $[0,+\infty)$, and it belongs to ${\cal B}{\cal T}{\cal T}$.\end{lemma} 
\begin{proof}
Just like in the proof of Proposition $\ref{sol-real-Jordan}$, we 
note that system~\ref{flow-X^2} can be solved recursively, 
starting with $X_1(t)$, followed by $X_2(t),X_3(t),\dots$. Only the first equation for $X_1$ is nonlinear, whereas 
the equations for $X_2,X_3,\dots$ are linear. To see this, we write these equations:
\begin{equation*}\label{recursion2}
{\dot X_i=\begin{cases}X_1^2,\;\textrm{ if }i=1\\
(2X_1(t))X_2,\;\textrm{ if }i=2\\
    (2X_1(t))X_i+\sum_{k=2}^{i-1}X_k(t)X_{i-(k-1)}(t),\; \textrm{ if }i>2
\end{cases}}.
\end{equation*}
with $X_1(0)=\alpha I_2 + \beta J_2$, $X_2(0)=I_2$ and $X_i(0)=0$ for $i > 2$.
Here we have used the fact that $X_1X_i+X_iX_1=2X_1X_i$, since any two matrices
of the form $p I_2 + q J_2$ commute and the matrices $X_i(t)$ are of this form. 

By Lemma $\ref{2by2-X^2}$, the forward solution for $X_1(t)$ is defined for all
$t$ in $[0,+\infty)$ (and in fact, converges to zero as $t\rightarrow +\infty$).

Since the $X_1(t)$ commute for every pair of $t$'s,  the forward solution for
$X_2(t)$ is given by \cite{Rugh1996} $X_2(t)=\exponent^{\int_0^t2X_1(s)ds}$, for
$t\in [0,+\infty),$ where this solution exists for all forward times $t$
because $X_1(t)$ is bounded and continuous. Similarly, the forward solution for
$X_i(t)$ when $i>2$, is given by the variation of constants formula: 
$$
X_i(t)=X_2(t)\left[
\int_0^tX_2^{-1}(s)\left(\sum_{k=2}^{i-1}X_k(s)X_{i-(k-1)}(s)\right)ds\right],$$
for $t\in [0,+\infty)\textrm{ when }i>2,$ where these solutions are recursively
defined for all forward times because the formula only involves integrals of
continuous functions.
\end{proof}

Combining both results, puts us in a position to state and prove our main
result.
\begin{theorem}\label{main2-X^2}
Let $X(0)\in \reals^{n\times n}$ and $(T,\Lambda_0)$ as in~\ref{jor} with~\ref{block-jor}. Let 
$a_1>a_2\geq \dots \geq a_k$ with $a_1>0$ a simple eigenvalue with corresponding right and left-eigenvectors 
$U_1$ and $V_1^T$ respectively:
$$
X(0)U_1=a_1U_1\textrm{ and }V_1^TX(0)=a_1V_1^T.
$$
Then the forward solution $X(t)$ of~\ref{flow-X^2} is defined for $[0,1/a_1)$, and
$$
\lim_{t\rightarrow 1/a_1}\frac{X(t)}{|X(t)|_F}=\frac{U_1V_1^T}{|U_1V_1^T|_F}.
$$
\end{theorem}
\begin{proof}
Consider the initial value problem ${\dot
\Lambda}=\Lambda^2$ with $\Lambda(0)=\Lambda_0$, whose solution is given by
$$\Lambda(t)=\diag(A_1(t), \ldots, A_k(t), B_1(t), \ldots, B_l(t)),$$ where for
all $i=1,\dots,k$, $A_i(t)$ is the forward solution of~\ref{flow-X^2} with
$A_i(0)$ of the form $A_i$ in ~\ref{block-jor}, which by Lemma
$\ref{sol-real-Jordan}$ is defined for all $t\in [0,1/a_i)$.  Since $a_1>a_2\geq
  \dots \geq a_k$, $A_1(t)$ blows up first when $t\rightarrow 1/a_1$.  The
  matrices $B_j(t)$, $j=1,\dots,l$, are the forward solution of~\ref{flow-X^2}
  with $B_j(0)$ of the form $B_j$ in ~\ref{block-jor}, and by Lemma
  $\ref{sol-complex-Jordan}$, they are defined for all $t$ in $[0,+\infty)$.

This clearly shows that $\Lambda (t)$ is defined in forward time for $t$ in
$[0,1/a_1)$. Since the solution of~\ref{flow-X^2} is given by
$X(t)=T\Lambda(t)T^{-1}$, $X(t)$ is also defined in forward time for $t$ in
$[0,1/a_1)$, and it follows  that 
\begin{align*}
\lim_{t\rightarrow 1/a_1}\frac{X(t)}{|X(t)|_F}
  &=\lim_{t\rightarrow 1/a_1}\frac{T\Lambda(t)T^{-1}}{|X(t)|_F}\\
  &=\frac{Te_1e_1^TT^{-1}}{|Te_1e_1^TT^{-1}|_F}
   =\frac{U_1V_1^T}{|U_1V_1^T|_F},
\end{align*}
where $e_1$ denotes the first standard unit basis vector of $\reals^n$.
\end{proof}

Theorem $\ref{main2-X^2}$ implies that social balance is usually not achieved
when $X(0)$ is an arbitrary real initial condition. Indeed, if $X_0$ has a
simple, positive, real eigenvalue $a_1$, and if we assume that no entry of the
right and left eigenvectors $U_1$ and $V_1^T$ are zero (an assumption which is
generically satisfied), then in general, up to a permutation of its entries, the
sign patterns of $U_1$ and $V_1^T$ are:
$$
U_1=\begin{pmatrix}
+\\
+\\
\hline
-\\
-
\end{pmatrix}\textrm{ and } V_1^T=\begin{pmatrix}
+&-&\vline &+&-
\end{pmatrix}
$$
implies that
$$
U_1V_1^T=\begin{pmatrix}
+& -&\vline& +& - \\
+& -&\vline& +& - \\
\hline
-& +&\vline& -& + \\
-& +&\vline& -& + \\
\end{pmatrix}.
$$
Then Theorem $\ref{factions}$ implies that the normalized state of the system does not become balanced in finite time.

This shows that in general, unless $X_0$ is normal (so that Theorem
$\ref{main-X^2}$ is applicable), we cannot expect that social balance will
emerge for system~\ref{flow-X^2}.

\section{Equation $\dot{X}=XX^T$}
We now consider
\begin{equation}\label{flow-XX^T}
{\dot X}=XX^T, X(0)=X_0,
\end{equation}
where again, each $X_{ij}$ denotes the real-valued opinion agent $i$ has about
agent $j$. As before, for $i=j$, the value of  $X_{ii}$ is interpreted as a
measure of self-esteem of agent $i$.  We can also write the equations entrywise:
\begin{equation}\label{entry-wise}
{\dot X}_{ij}=\sum_k X_{ik}X_{jk}.
\end{equation}

As in the case of model ${\dot X}=X^2$, we split up the analysis in two parts.
First we consider system~\ref{flow-XX^T} with normal initial condition
$X_0$, and we shall see that not all initial conditions lead to the emergence of
a balanced network in this case, in contrast to the behaviour of
~\ref{flow-X^2}.  Secondly, we will see that for non-normal, generic initial
conditions $X_0$, we typically do get the emergence of social balance, also
contrasting the behaviour of~\ref{flow-X^2}.

\subsection*{Normal initial condition}
As for the model $\dot{X}=X^2$ the set ${\cal N}$ is invariant for system
~\ref{flow-XX^T}. By using the same diagonalisation as in Eq.
\ref{diagonalizable}, if $\Lambda(t)$ is the solution to the initial value
problem ${\dot \Lambda}=\Lambda\Lambda^T$, $\Lambda(0)=\Lambda_0$, then
$X(t):=U\Lambda(t)U^T$, is the solution to Eq.~\ref{flow-XX^T}. This shows it is
sufficient to solve system~\ref{flow-XX^T} in case of scalar $X$ or in case
of a specific $2\times 2$ normal matrix $X$. The scalar case is easy to solve
and follows Eq.~\ref{equ:scalar}, so we turn to the $2\times 2$ case by
considering
\begin{equation}\label{flow2}
{\dot X}=XX^T,\; X(0)=\alpha I_2 + \beta J_2,
\textrm{ where } \beta \neq 0.
\end{equation}
We define the angle $\phi$ as
\begin{equation}\label{angle}
\phi=\arctan\left(\frac{\alpha}{\beta} \right),\; \phi\in\left(-\frac{\pi}{2},\frac{\pi}{2}\right).
\end{equation}

\begin{lemma}\label{2by2}
Define ${\bar t}$ as 
\begin{equation}\label{escape}
{\bar t}=\frac{\pi}{2\beta}-\frac{\phi}{\beta}.
\end{equation}
Then the forward solution $X(t)$ of~\ref{flow2} is:
\begin{equation}\label{2by2-sol}
X(t)=\beta\tan(\beta t+\phi)I_2 + \beta J_2,\;\; t\in [0, {\bar t}).
\end{equation}
Moreover,
$$
\lim_{t\rightarrow {\bar t}-} X(t)=+\infty I_2 + \beta J_2 \textrm{ and }\lim_{t\rightarrow {\bar t}-} \frac{X(t)}{|X(t)|_F}=\frac{\sqrt{2}}{2}I_2.
$$
\end{lemma}
\begin{proof}
Let $X_0=S_0+A_0$, $S_0=\alpha I_2$, and $A_0=\beta J_2$. Then the solution
$X(t)$ of~\ref{flow2} can be decomposed as $S(t)+A(t)$, where
\begin{align}
{\dot S}&=(S+A)(S-A),\;\; S(0)=S_0,\\
{\dot A}&=0,\;\; A(0)=A_0,
\end{align}
so $A(t) = A_0$, and reduces to
\begin{equation}
{\dot S}=(S+A_0)(S-A_0),\;\; S(0)=S_0\label{s-part}\\
\end{equation}
Note that~\ref{s-part} is a matrix Riccati differential equation with the
property that the line ${\cal L}=\{\alpha I_2| \alpha \in \reals\}$, is an
invariant set under the flow. Therefore it suffices to solve the scalar Riccati
differential equation corresponding to the dynamics of the diagonal entries of
$S$: ${\dot s}=s^2+\beta^2,\; s(0)=\alpha,$
whose forward solution is:
$s(t)=\beta \tan\left(\beta t+\phi \right),$ for $t\in (0,{\bar t}),$
where ${\bar t}$ is given by~\ref{escape}.
Consequently, the forward solution $X(t)$ of~\ref{flow2} is given by:
$X(t)=S(t)+A_0=\beta\tan(\beta t+\phi)I_2 + \beta J_2,$
for $t\in (0,{\bar t}),$ and thus
$\lim_{t\rightarrow {\bar t}-}X(t)=+\infty I_2 + \beta J_2$
and 
$$\lim_{t\rightarrow {\bar t}-}\frac{X(t)}{|X(t)|_F}=\frac{X(t)}{\sqrt{2}|\beta \sec(\beta t+\phi)|}=\frac{\sqrt{2}}{2}I_2.$$
\end{proof}

Combining the solution for the $1\times1$ scalar case in Eq.~\ref{equ:scalar}
and Lemma $\ref{2by2}$ yields our main result:
\begin{theorem}\label{main-XX^T}
Let $X_0\in {\cal N}$, and let $(U,\Lambda_0)$ be as in Lemma $\ref{diagonalizable}$. Define
$$
{\bar t}_i=\begin{cases}
1/a_i\textrm{ if } a_i>0\\
+\infty \textrm{ if } a_i\leq 0
\end{cases}
\textrm{ for all } i=1,\dots,k, 
$$
and
$$
{\bar t}_j=\frac{\pi}{2\beta_j}-\frac{\phi_j}{\beta_j} \textrm{ for all } j=1,\dots,l, 
$$
where $\phi_j=\arctan\left( \frac{\alpha_j}{\beta_j} \right)$ and let ${\bar
t}=\min_{i,j}\{{\bar t}_i,{\bar t}_j \}$.  Then the forward solution $X(t)$ of
~\ref{flow-XX^T} is defined for $[0,{\bar t})$. 

If there is a unique $i^*\in \{1,\dots,k\}$ such that ${\bar t}={\bar t}_{i^*}$
is finite, then
$$\lim_{t\rightarrow {\bar t}_{i^*}-}\frac{X(t)}{|X(t)|_F}=U_{i^*}U_{i^*}^T,$$
where $U_{i^*}$ is the $i^*$th column of $U$, an eigenvector corresponding to
eigenvalue $a_{i^*}$ of $X_0$. 

If there is a unique $j^*\in \{1,\dots, l\}$ such that ${\bar t}={\bar t}_{j^*}$, then
$$
\lim_{t\rightarrow {\bar t}_{j^*}-}\frac{X(t)}{|X(t)|_F}=\frac{\sqrt{2}}{2}U_{j^*}U_{j^*}^T,
$$
where $U_{j^*}$ is an $n\times 2$ matrix consisting of the two consecutive columns of $U$ which correspond 
to the columns of the $2\times 2$ block $B_{j^*}$ in $\Lambda_0$.
\end{theorem}
\begin{proof}
Consider the initial value problem:
$$
{\dot \Lambda}=\Lambda \Lambda^T,\; \Lambda(0)=\Lambda_0.
$$
By Lemma $\ref{2by2}$ its solution is given by
$$
\Lambda(t)=\begin{pmatrix}
\frac{a_1}{1-a_1t}&\dots&0 &0&\dots& 0\\
\vdots&  \ddots &\vdots& \vdots&\ddots&\vdots\\
0&\dots& \frac{a_k}{1-a_kt}&0&\dots&0\\
0&\dots&0& X_1(t)&\dots&0\\
\vdots&\ddots&\vdots&\vdots&\ddots&\vdots\\
0&\dots&0&0&\dots&X_l(t)
\end{pmatrix},
$$
where for all $j=1,\dots,l$, $X_j(t)$ is given by the $2\times 2$ matrix in
~\ref{2by2-sol} with $\beta$, $\phi$ and ${\bar t}$ replaced by $\beta_j$,
$\phi_j$ and ${\bar t}_j$ respectively. This clearly shows that $\Lambda (t)$ is
defined in forward time for $t$ in $[0,{\bar t})$. Since the solution of
~\ref{flow-XX^T} is given by $X(t)=U\Lambda(t)U^T$, $X(t)$ is also defined in
forward time for $t$ in $[0,{\bar t})$.  It follows from~\ref{norm-invar}
that
$$
\frac{X(t)}{|X(t)|_F}=U\frac{\Lambda(t)}{|\Lambda (t)|_F}U^T.
$$
If $i^*\in\{1,\dots,k\}$ is the unique value such that ${\bar t}={\bar t}_{i^*}$, then
\begin{align*}
\lim_{t\rightarrow {\bar t}_i^*}\frac{X(t)}{|X(t)|_F}
  &=U\lim_{t\rightarrow {\bar t}_i^*}\frac{\Lambda(t)}{|\Lambda (t)|_F}U^T\\
  &= Ue_{i^*}e_{i^*}^TU^T=U_{i^*}U_{i^*}^T,
\end{align*}
where $e_{i^*}$ denotes the $i^*$th standard unit basis vector of $\reals^n$.

If $j^*\in \{1,\dots,l\}$  is the unique value such that ${\bar t}={\bar t}_{j^*}$, then by Lemma $\ref{2by2}$:
\begin{align*}
\lim_{t\rightarrow {\bar t}_j^*}\frac{X(t)}{|X(t)|_F}
  &=U\lim_{t\rightarrow {\bar t}_j^*}\frac{\Lambda(t)}{|\Lambda (t)|_F}U^T\\
  &=\frac{\sqrt{2}}{2}UE_{j^*}U^T=\frac{\sqrt{2}}{2}U_{j^*}U_{j^*}^T,
\end{align*}
where $E_{j^*}$ has exactly two non-zero entries equal to $1$ on the diagonal positions corresponding 
to the block $B_{j^*}$ in $\Lambda_0$.
\end{proof}

A particular consequence of Theorem $\ref{main-XX^T}$ is that if $X_0$ has a
complex pair of eigenvalues, the solution of~\ref{flow-XX^T} always blows up
in finite time, even if all real eigenvalues of $X_0$ are non-positive. Recall
that the solution of~\ref{flow-X^2} blows up in finite time, if and only if
$X_0$ has a positive, real eigenvalue. Another implication of Theorem
$\ref{main-XX^T}$ is that if blow-up occurs, it may be due to a real eigenvalue
of $X_0$, or to a complex eigenvalue. In contrast, if the solution of
~\ref{flow-X^2} blows up in finite time, it is necessarily due to a positive,
real eigenvalue, and never to a complex eigenvalue. When the solution of
~\ref{flow-XX^T} blows up because of a positive, real eigenvalue of $X_0$,
the system will achieve balance, just as in the case of system
~\ref{flow-X^2}. If on the other hand, finite time blow up of
~\ref{flow-XX^T} is caused by a complex eigenvalue of $X_0$, we show that in
general one cannot expect to achieve a balanced network. Assume there is a
unique $j^*$ such that:
$$ \lim_{t\rightarrow {\bar t}_j^* -}\frac{X(t)}{|X(t)|_F}=\frac{\sqrt{2}}{2}U_{j^*}U_{j^*}^T,$$
Assuming that no entry of $U_{j^*}$ is zero, the sign pattern of
$U_{j^*}U_{j^*}^T$, with 
$$U_j^* = 
\begin{pmatrix}
p_1& q_1 \\
p_2& -q_2 \\
-p_3& q_3 \\
-p_4& -q_4
\end{pmatrix}$$
is given by:
$$
\begin{pmatrix}
+&?&?&- \\
?&+&-&? \\
?&-&+&? \\
-&?&?&+
\end{pmatrix},
$$
up to a suitable permutation, where all $p_i$ and $q_i$, $i=1,\dots,4$, are
entrywise positive vectors, and where
$$
\langle p_1,q_1 \rangle + \langle p_4,q_4 \rangle = \langle p_2,q_2 \rangle +
\langle p_3,q_3 \rangle ,
$$
because $U$ is an orthogonal matrix. The $?$ are not entirely arbitrary because
$U_{j^*}U_{j^*}^T$ is a symmetric matrix, but besides that their signs can be
arbitrary.

\subsection*{Generic initial condition}
Consider
\begin{equation}\label{flow}
{\dot X}=XX^T, X(0)=X_0,
\end{equation}
where $X$ is a real $n\times n$ matrix, which is not necessarily normal.

We first decompose the flow~\ref{flow} into  flows for the symmetric and
skew-symmetric parts of $X$.  Let $X=S+A$,$X_0=S_0+A_0$, where $S,S_0\in {\cal
S}$ and $A,A_0\in {\cal A}$ are the unique symmetric and skew-symmetric parts of
$X$ and $X_0$ respectively.  If $X(t)$ satisfies~\ref{flow}, then it can be
verified that $S(t)$ and $A(t)$ satisfy the system:
\begin{align}
{\dot S}&=(S+A)(S-A),\; S(0)=0,\label{sym-flow}\\
{\dot A}&=0,\; A(0)=A_0,\label{A-part}
\end{align}
Consequently, $A(t)=A_0$ for all $t$, and thus the skew-symmetric part of the
solution $X(t)$ of~\ref{flow} remains constant and equal to $A_0$.
Throughout this subsection we assume that $A_0\neq 0$, for otherwise $X(0)$ is
symmetric, hence normal, and the results from the previous subsection apply.  It
follows that we only need to understand the dynamics of the symmetric part.
Then the solution $X(t)$ to~\ref{flow} is given by $X(t)=S(t)+A_0$, where
$S(t)$ solves~\ref{sym-flow}, and in view of ~\ref{perp}, there follows
by Pythagoras' Theorem that: $|X(t)|_F^2=|S(t)|_F^2+|A_0|_F^2$, and thus 
\begin{equation}\label{pyt}
\frac{X(t)}{|X(t)|_F}=\frac{S(t)+A_0}{\left( |S(t)|_F^2+|A_0|_F^2 \right)^{\frac{1}{2}}}.
\end{equation}

Next we shall derive an explicit expression for the solution $S(t)$ of
~\ref{sym-flow}.  We start by performing a change of variables:
\begin{equation}\label{S-hat}
{\hat S}(t)=\exponent^{-tA_0}S(t)\exponent^{tA_0}.
\end{equation}
This yields the equation
\begin{equation}\label{hat-eqn}
{\dot {\hat S}}={\hat S}^2-A_0^2,\;\; {\hat S}(0)=S_0.
\end{equation}
We perform a further transformation which diagonalizes $-A_0^2$: Let $V$ be an
orthogonal matrix such that $-V^TA_0^2V=D^2,$ where
$D:=\diag(0,\omega_1 I_2,\ldots,\omega_k I_k)$
where $k\geq 1$ (because $A_0 \neq 0$) and all $\omega_j> 0$ without
loss of generality. Setting 
\begin{equation}\label{S-tilde}
{\tilde S}=V^T{\hat S}V,
\end{equation} 
and multiplying equation~\ref{hat-eqn} 
by $V$ on the left, and by $V^T$ on the right, we find that:
\begin{equation}\label{tilde-eqn}
{\dot {\tilde S}}={\tilde S}^2+D^2,\;\; {\tilde S}(0)={\tilde S_0}:=V^TS_0V.
\end{equation}
Notice that this is a matrix Riccati differential equation, a 
class of equations with specific properties which  are briefly reviewed next. 

Consider a general matrix Riccati differential equation:
\begin{equation}\label{RDE}
{\dot S}=SMS-SL-L^TS+N,
\end{equation}
where $M=M^T$,$N=N^T$ and $L$ arbitrary, defined on ${\cal S}$. Associated to
this equation is a linear system
\begin{equation}\label{Ham-sys}
\begin{pmatrix}
{\dot P}\\
{\dot Q}
\end{pmatrix}
=
H
\begin{pmatrix}
P\\
Q
\end{pmatrix},\;\; H:=\begin{pmatrix}
L& -M\\
N&-L^T
\end{pmatrix},
\end{equation}
where $H$ is a Hamiltonian matrix, i.e. $J_{2n}H=(J_{2n}H)^T$ holds, where
$J_{2n}$ is as defined in Eq.~\ref{skew_J}.
The following fact is well-known.
\begin{lemma}\label{ratio-lemma}
  Let $\left( \begin{smallmatrix} P(t)\\ Q(t) \end{smallmatrix} \right)$ be a
  solution of~\ref{Ham-sys}. Then, provided that $P(t)$ is non-singular,
  \begin{equation}\label{ratio}
  S(t)=Q(t)P(t)^{-1},
  \end{equation}
  is a solution of~\ref{RDE}.  Conversely, if $S(t)$ is a solution of
 ~\ref{RDE}, then there exists a solution 
  $\left(\begin{smallmatrix}P(t)\\ Q(t) \end{smallmatrix} \right)$ 
  of~\ref{Ham-sys} such that~\ref{ratio} holds, provided that $P(t)$ is
  non-singular.
\end{lemma}
\begin{proof}
Taking derivatives in $S(t)P(t)=Q(t)$ yields that ${\dot S}=({\dot Q}-S{\dot P})P^{-1}$, and using~\ref{Ham-sys},
$${\dot S}=(NP-L^TQ-S(LP-MQ))P^{-1}=N-L^TS-SL+SMS,$$
showing that $S(t)$ solves~\ref{RDE}. For the converse, let $S(t)$ be a
solution of~\ref{RDE}. Let  
$\left( \begin{smallmatrix}P(t)\\ Q(t) \end{smallmatrix} \right)$ 
with
$\left( \begin{smallmatrix}P(0)\\ Q(0) \end{smallmatrix} \right)=
  \left( \begin{smallmatrix}I_n\\ S(0) \end{smallmatrix} \right)$ 
be the solution of~\ref{Ham-sys}. Then
\begin{align*}
&\frac{d}{dt}\left(Q(t)P^{-1}(t)\right)\\
=&{\dot Q}P^{-1}-QP^{-1}{\dot P}P^{-1}\\
=&(NP-L^TQ)P^{-1}-QP^{-1}(LP-MQ)P^{-1}\\
=&(QP^{-1})M(QP^{-1})-(QP^{-1})L-L^T(QP^{-1})+N,
\end{align*}
which implies that $QP^{-1}$ is a solution to~\ref{RDE}. Since
$S(0)=Q(0)P^{-1}(0)$, it follows from uniqueness of solutions that
$S(t)=Q(t)P^{-1}(t)$.
\end{proof}
In other words, in principle we can solve the nonlinear equation~\ref{RDE}
by first solving the linear system~\ref{Ham-sys}, and then use formula
~\ref{ratio} to determine the solution of~\ref{RDE}. 

We carry this out for our particular Riccati equation~\ref{tilde-eqn} which
is of the form~\ref{RDE} if
$M=I_n,\;\; L=0,\;\; N=D^2$.
The corresponding Hamiltonian is
$H=\left( \begin{smallmatrix}
0& -I_n\\
D^2& 0
\end{smallmatrix} \right)$.
We partition $D$ in singular and non-singular parts:
$$
D=\begin{pmatrix}
0& 0\\
0& {\tilde D}
\end{pmatrix},\textrm{ where }{\tilde D}:=\begin{pmatrix}\omega_1 I_2 &\dots&0\\
\vdots&\ddots&\vdots\\
0&\dots&\omega_kI_2
\end{pmatrix},
$$
where ${\tilde D}$ is positive definite since all $\omega_j > 0$. Partitioning $H$ correspondingly:
\begin{equation}\label{Ham-ex}
H=\begin{pmatrix}
0 & 0 & \vline& -I_{n-2k}& 0 \\
0&  0&  \vline&0& -I_{2k}\\
\hline 
 0& 0&\vline& 0& 0\\
  0& {\tilde D}^2&\vline& 0&  0
\end{pmatrix}.
\end{equation}
This matrix is then exponentiated to solve system~\ref{Ham-sys}:
$$
\begin{pmatrix}
P(t)\\
\hline
Q(t)
\end{pmatrix}=
\begin{pmatrix}
I_{n-2k}&0&\vline&-tI_{n-2k}&0\\
0&c&\vline&0&-{\tilde D}^{-1}s\\
\hline
0&0&\vline&I_{n-2k}&0\\
0&{\tilde D}s&\vline&0&c
\end{pmatrix}
\begin{pmatrix}
P(0)\\
\hline
Q(0)
\end{pmatrix},
$$
where we have introduced the following notation:
$$s(t):=\diag(\sin(\omega_1 t) I_2, \ldots, \sin(\omega_k t) I_2)
     =\sin({\tilde D}t),$$
and similarly $c(t)=\cos({\tilde D}t)$.  By setting $P(0)=I_n$, and
$Q(0)={\tilde S}_0$, and using Lemma $\ref{ratio-lemma}$, it follows that the
solution of the initial value problem~\ref{tilde-eqn} is given by ${\tilde
S}(t)=Q(t)P(t)^{-1}$, 
\begin{equation}\label{sol-flow}
\begin{pmatrix}
P(t)\\
\hline 
Q(t)
\end{pmatrix}=
\begin{pmatrix}
\begin{pmatrix}
  (I_{n-2k} - t)\tilde{S}_0 &0\\
  0& c(t) - {\tilde D}^{-1}s(t){\tilde S_0}\\
\end{pmatrix}\\
\hline
\begin{pmatrix}
I_{n-2k}{\tilde S_0}&0\\
0& {\tilde D}s(t) + c(t){\tilde S_0}
\end{pmatrix}
\end{pmatrix},
\end{equation}
for all $t$ for which $P(t)$ is non-singular. We now make the following
assumption:

\begin{assumption}
The matrix $P(t)$ is non-singular for all $t$ in $[0,{\bar t})$, where ${\bar
t}$ is finite and such that $s(t)$ is non-singular for all $t$ in $(0,{\bar
t})$. Moreover, $P({\bar t})$ has rank $n-1$, or equivalently, has a simple
eigenvalue at zero.
\label{assum}
\end{assumption}

Later we will show that this assumption is generically satisfied, and also that
\begin{equation}
   t^*={\bar t},
   \label{equalt}
\end{equation}
where $[0,t^*)$ is the maximal forward interval of existence of
the solution ${\tilde S}(t)$ of the initial value problem~\ref{tilde-eqn}.
Consequently, the theory of ODE's implies that $\lim_{t\rightarrow {\bar t}}|{\tilde S}(t)|_F=+\infty$, i.e. that 
${\bar t}$ is the blow-up time for the solution ${\tilde S}(t)$.

Assuming for the moment that assumption~\ref{assum} is satisfied,
back-transformation using~\ref{S-hat} and~\ref{S-tilde}, yields that the
solution $S(t)$ of~\ref{sym-flow} is $S(t)=\exponent^{tA_0}V{\tilde
S}(t)V^T\exponent^{-tA_0}$, which is defined for all $t$ in $[0,{\bar t})$,
because $\exponent^{tA_0}V$ is bounded for all $t$ (as it is an orthogonal
matrix). It follows from~\ref{norm-invar} that
\begin{equation}\label{intermed1}
\lim_{t\rightarrow {\bar t}}\frac{S(t)}{|S(t)|_F}=\exponent^{{\bar t}A_0}V
\left(\lim_{t\rightarrow {\bar t}} \frac{{\tilde S}(t)}{|{\tilde S}(t)|_F}\right)
V^T\exponent^{-{\bar t}A_0},
\end{equation}
provided that at least one of the two limits exists. 
Partitioning ${\tilde S_0}$ in~\ref{sol-flow} as follows:
$$
{\tilde S_0}=\begin{pmatrix}({\tilde S_0})_{11}&({\tilde S_0})_{12}\\ 
({\tilde S_0})_{12}^T&({\tilde S_0})_{22}
 \end{pmatrix},\textrm{ with } 
 \begin{array}{c}
   ({\tilde S_0})_{11}= ({\tilde S_0})_{11}^T \\
   ({\tilde S_0})_{22}= ({\tilde S_0})_{22}^T
 \end{array},
$$
we can rewrite $P(t)$ and $Q(t)$ on the time interval $(0,{\bar t})$ as:
$P(t)=\Delta (t) M(t)$ with,
$$\Delta(t)=
\begin{pmatrix}
tI_{n-2k}&0\\
0&{\tilde D}^{-1}s(t)
\end{pmatrix},$$
and
$$M(t)=\begin{pmatrix}
1/t-({\tilde S_0})_{11}& -({\tilde S_0})_{12}\\
-({\tilde S_0})_{12}^T& {\tilde D}c(t)s^{-1}(t)-({\tilde S_0})_{22}
\end{pmatrix}=M^T(t), $$
and
$$Q(t)=\begin{pmatrix}
({\tilde S_0})_{11}&({\tilde S_0})_{12}\\
c(t)({\tilde S_0})_{12}^T& {\tilde D}s(t)+c(t)({\tilde S_0})_{22}
\end{pmatrix}. $$
Note that the factorization of $P(t)$ is well-defined on $(0,{\bar t})$ because
by assumption~\ref{assum}, the matrix $s(t)$ is non-singular in the
interval $(0,{\bar t})$. Moreover, assumption~\ref{assum} also
implies there exists a nonzero vector $u$ corresponding to the zero eigenvalue
of $M({\bar t})$, i.e. $M({\bar t})u=0$, and that $u$ is uniquely defined up to
scalar multiplication because the zero eigenvalue is simple. More explicitly,
partitioning $u$ as 
$\left( \begin{smallmatrix}u_1\\u_2 \end{smallmatrix} \right)$, 
there holds that
\begin{equation}\label{explicit}
\begin{pmatrix}
1/{\bar t}-({\tilde S_0})_{11}& -({\tilde S_0})_{12}\\
-({\tilde S_0})_{12}^T&{\tilde D}c({\bar t})s^{-1}({\bar t})-({\tilde S_0})_{22}
\end{pmatrix}\begin{pmatrix}u_1\\u_2 \end{pmatrix}=0.
\end{equation}
Notice that $M(t)$ is at least real-analytic on the interval $(0,{\bar t})$. Hence, it follows from
\cite{Kato1995} (see also \cite{Bunse-Gerstner1991,Still2001}), that there is 
an orthogonal matrix $U(t)$, and a diagonal matrix $\Lambda(t)$, both real-analytic on $(0,{\bar t})$, such that:
$M(t)=U(t)\Lambda(t)U^T(t)$, for $t\in (0,{\bar t}),$ and thus
$M^{-1}(t)=U(t)\Lambda^{-1}(t)U^T(t)$, for $t\in (0,{\bar t})$.
Returning to ~\ref{intermed1}, we obtain that:
\begin{align*}
  & \lim_{t\rightarrow {\bar t}}\frac{S(t)}{|S(t)|_F} \\
  =&
  \exponent^{{\bar t}A_0}V
  \lim_{t\rightarrow {\bar t}}\frac{Q(t)U(t)\Lambda^{-1}(t)U^T(t)\Delta^{-1}(t)}{|Q(t)U(t)\Lambda^{-1}(t)U^T(t)\Delta^{-1}(t)|_F}
  V^T\exponent^{-{\bar t}A_0} \\
  =& \exponent^{{\bar t}A_0}V
  \frac{Q({\bar t})uu^T\Delta^{-1}(t)}{|Q({\bar t})uu^T\Delta^{-1}(t)|_F}
  V^T\exponent^{-{\bar t}A_0}.
\end{align*}
Here, we have used the fact that $M^{-1}(t)$ is positive definite on the
interval $(0,{\bar t})$, so that its largest eigenvalue (which is simple for all
$t<{\bar t}$ sufficiently close to ${\bar t}$, because of assumption
\ref{assum} approaches $+\infty$ -and not $-\infty$- as
$t\rightarrow {\bar t}$.  To see this, note that from its definition follows that $M(t)$ is positive definite for
all sufficiently small $t>0$, because ${\tilde D}$ is positive definite. Moreover, $M(t)$ is non-singular on $(0,{\bar t})$
since by assumption $({\bf A})$, $P(t)$ is non-singular on $(0,{\bar t})$, and
because $M(t)=\Delta^{-1}(t)P(t)$ (it is clear from its definition and
assumption~\ref{assum} that $\Delta(t)$ is non-singular on $(0,{\bar
t})$ as well). Consequently, the smallest eigenvalue of $M(t)$ remains positive
in $(0,{\bar t})$, as it approaches zero as $t\rightarrow {\bar t}$. This implies
that the largest eigenvalue of $M^{-1}(t)$ is positive on $(0,{\bar t})$, and
approaches $+\infty$ as $t\rightarrow {\bar t}$, as claimed.

Note that:
\begin{align*}
Q({\bar t})u&=\begin{pmatrix}
({\tilde S_0})_{11}&({\tilde S_0})_{12}\\
c({\bar t})({\tilde S_0})^T_{12}&{\tilde D}s({\bar t})+c({\bar t})({\tilde S_0})_{22}
\end{pmatrix}\begin{pmatrix}u_1\\u_2\end{pmatrix}\\
  &= \begin{pmatrix}
(1/{\bar t})u_1\\
{\tilde D}s^{-1}({\bar t})u_2
\end{pmatrix}=\Delta^{-1}({\bar t})u,
\end{align*}
where in the second equality, we used the second row of  ~\ref{explicit} , multiplied by $c({\bar t})$.
From this follows that 
\begin{align*}
  \lim_{t\rightarrow {\bar t}}\frac{S(t)}{|S(t)|_F} &=
  \exponent^{{\bar t}A_0}V
  \frac{\Delta^{-1}({\bar t})uu^T\Delta^{-1}({\bar t})}{|\Delta^{-1}({\bar t})uu^T\Delta^{-1}({\bar t})|_F}
  V^T\exponent^{-{\bar t}A_0} \\
  &=
  \frac{ww^T}{|ww^T|_F},
\end{align*}
where $w := e^{\bar{t} A_0} V \Delta^{-1}({\bar t})u$.

Taking limits for $t\rightarrow {\bar t}$ in ~\ref{pyt}, and using the above equality, we finally arrive at the following
result, which implies that system~\ref{flow} evolves to a socially balanced
state (in normalized sense) when $t\rightarrow {\bar t}$:
\begin{proposition}\label{main-proposition}
Suppose that assumption~\ref{assum} holds and $A_0 \neq 0$. Then the solution 
$X(t)$ of~\ref{flow} satisfies:
$$
\lim_{t\rightarrow {\bar t}}\frac{X(t)}{|X(t)|_F}=\frac{ww^T}{|ww^T|_F}.
$$
with $w = e^{\bar{t} A_0} V \Delta^{-1}({\bar t})u$.
\end{proposition}

\subsection*{Genericity}

Generically, assumption~\ref{assum} holds, and~\ref{equalt} holds
as well. There are two aspects to assumption~\ref{assum}:
\begin{enumerate}
\item
The matrix $P(t)$ is nonsingular in the interval $[0,{\bar t})$, but singular at some finite ${\bar t}$ such that:
\begin{equation}\label{vlug}
{\bar t}<\min_{j=1,\dots,k}\frac{\pi}{\omega_j}.
\end{equation}
\item $P({\bar t})$ has a simple zero eigenvalue.
\end{enumerate}
To deal with the first item, suppose that the solution ${\tilde S}(t)$ of~\ref{tilde-eqn} is defined for all $t\in [0,t^*)$ for some finite positive $t^*$. By Lemma $\ref{ratio-lemma}$, there exist $P(t)$ and $Q(t)$ such that 
${\tilde S}(t)=Q(t)P^{-1}(t)$, where $P(t)$ and $Q(t)$ are components of the solution of system~\ref{Ham-sys} 
with $H$ defined in~\ref{Ham-ex}. 
Then necessarily ${\bar t}\leq t^*$. Thus, if we can show that $t^*<\min_j \pi/\omega_j$, then~\ref{vlug} holds.
To show that $t^*<\min_j \pi/\omega_j$, we rely on a particular property of matrix Riccati differential equations 
~\ref{RDE}: their solutions preserve the order generated by the cone of non-negative symmetric matrices, see 
\cite{DeLeenheer2004}. More 
precisely, if $S_1(t)$ and $S_2(t)$ are solutions of~\ref{RDE}, and if 
$S_1(0)\preceq S_2(0),$
then $S_1(t)\preceq S_2(t),$ for all $t \geq 0$ for which both solutions are
defined. The partial order notation $S_1(t)\preceq S_2(t)$ means that the
difference $S_2(t)-S_1(t)$ is a positive semi-definite matrix.

We apply this to equation~\ref{tilde-eqn} with ${\tilde
S_1}(0)=\alpha_{\min} I_n$ and ${\tilde S_2}(0)={\tilde S}(0)$,
where we choose $\alpha_{\min}$ as the smallest eigenvalue of ${\tilde S}(0)$
(or equivalently, of $S(0)=S_0$, since ${\tilde S}(0)=V^TS_0V$), so that
clearly ${\tilde S_1}(0)\preceq {\tilde S_2}(0)$.
Consequently, by the monotonicity property of system~\ref{tilde-eqn}, it
follows that ${\tilde S_1}(t)\preceq {\tilde S}(t)$, as long as both solutions
are defined. We can calculate the blow-up time $t_1^*$ of ${\tilde S_1}(t)$
explicitly, and then it follows that $t^*\leq t_1^*$, where $t^*$ is the blow-up
time of ${\tilde S}(t)$. Indeed, equations of system~\ref{tilde-eqn} decouple
for an initial condition of the form $\alpha_{\min}I_n$, and the resulting
scalar equations are scalar Riccati equations we have solved before. The blow-up
time for ${\tilde S_1}(t)$ is given by:
$$
t_1^*=\begin{cases}
\min_{j=1,\dots,k}\left(\frac{\pi}{2\omega_j}-\frac{\phi_j}{\omega_j}
  \right),\;\; \textrm{ if } \alpha_{\min}\leq 0\\
\min_{j=1,\dots,k}\left(\frac{1}{\alpha_{\min}},\frac{\pi}{2\omega_j}-\frac{\phi_j}{\omega_j}
  \right),\;\; \textrm{ if } \alpha_{\min}>0\
\end{cases}.
$$
with $\phi_j:=\arctan \left(\frac{\alpha_{\min}}{\omega_j} \right)\in
\left(-\frac{\pi}{2},\frac{\pi}{2} \right)$.
Notice that for all $j=1,\dots,k$, there holds that
$\frac{\pi}{2\omega_j}-\frac{\phi_j}{\omega_j}<\frac{\pi}{\omega_j},$ because by
definition, $\frac{\phi_j}{\omega_j}\in
(-\frac{\pi}{2\omega_j},\frac{\pi}{2\omega_j})$.
Consequently,
$$
{\bar t}\leq t^*\leq t_1^*<\min_{j=1,\dots,k} \frac{\pi}{\omega_j},
$$
which establishes ~\ref{vlug}. In other words, we have shown that the first item in assumption~\ref{assum} is always satisfied.

The second item in assumption~\ref{assum} may fail, but holds for
generic initial conditions as we show next. For this we first point out that the
derivative of each eigenvalue of $M(t)$ is a strictly decreasing function in the
interval $(0,{\bar t})$, independently of the value of the matrix $\tilde S_0$.
Indeed, the derivative
of eigenvalue $\lambda_j(t)$ of $M(t)$ equals (see \cite{Kato1995}) :
\begin{align*}
  {\dot\lambda}_j(t)&=u_j(t)^T\dot M(t)u_j(t) \\
  &= -u_j(t)^T \begin{pmatrix}
    1/t^2&0\\
    0&{\tilde D}^2s^{-2}(t)
  \end{pmatrix}u_j(t),
\end{align*}
where $u_j(t)$ is the normalized eigenvector of $M(t)$ corresponding to
$\lambda_j(t)$, and which is analytic in the considered interval.  Since $\dot
M(t)$ is negative definite in that interval, $\dot\lambda_j(t)$ is also negative
and hence all eigenvalues of $M(t)$ are strictly decreasing functions of $t$ in
that interval. Suppose now that $M(t)$ has a multiple eigenvalue 0 at $t={\bar
t}$, then $M({\bar t})$ is positive semi-definite since ${\bar t}$ is the first
singular point of $M(t)$ and the eigenvalues are decreasing function of $t$. If
we now choose a positive semi-definite $\Delta_{\tilde S_0}$ of nullity 1, such
that $M({\bar t})+\Delta_{\tilde S_0}$ also has nullity 1, then the perturbed
initial condition $({\tilde S}_0)_p=\tilde S_0-\Delta_{\tilde S_0}$ yields the
perturbed solution ${\tilde S}_p(t)$ which can be factored as
$Q_p(t)P^{-1}_p(t)$, and where $P_p(t)=\Delta(t)M_p(t)$ (note that $\Delta (t)$
remains the same as before the perturbation) for $M_p(t)=M(t)+\Delta_{\tilde
S_0}$ which now has a single root at the same minimal value ${\bar t}$. To
construct such a matrix $\Delta_{\tilde S_0}$ is simple since the only condition
it needs to satisfy is that $M({\bar t})$ and $\Delta_{\tilde S_0}$ have a
common null vector. Those degrees of freedom show that the second item in
assumption~\ref{assum} is indeed generic. 

Now that we have established that~\ref{assum} generically holds, we
show that ~\ref{equalt} is satisfied also.  The proof is by contradiction.
Earlier, we have shown that ${\bar t}\leq t^*$. Thus, if we suppose that
~\ref{equalt} fails, then necessarily ${\bar t}<t^*$.  This implies that
although $P({\bar t})$ is singular, the solution ${\tilde S}(t)$ exists for
$t={\bar t}$. Our goal is to show that $\lim_{t\rightarrow {\bar t}} |{\tilde
S}(t)|_F=+\infty$, which yields the desired contradiction (by the theory of ODE's).

We first claim the following:
\begin{equation}\label{claim}
\textrm{If } u\neq 0\textrm{ and } P({\bar t})u=0,\textrm{ then } Q({\bar t})u\neq 0.
\end{equation}
Indeed, if this were not the case, then there would exist some vector ${\bar
u}\neq 0$ such that $P({\bar t})\bar{u} = 0$ and $Q({\bar t})\bar{u} = 0$. On
the other hand, $P(t)$ and $Q(t)$ are components of the matrix product
$$
\begin{pmatrix}
P(t)\\
Q(t)
\end{pmatrix}=\exponent^{tH}\begin{pmatrix}
I_n\\
{\tilde S}_0
\end{pmatrix},
$$
where $H$ is defined in ~\ref{Ham-ex}. Multiplying the latter in $t={\bar t}$
by ${\bar u}$, and using the previous expression, it follows from the
invertibility of $\exponent^{{\bar t}H}$ that ${\bar u}=0$, a contradiction.
This establishes ~\ref{claim}.

In the previous section, we factored $P(t)$ as $P(t)=\Delta(t) M(t)$.  Since
$P(t)$ is non-singular on $[0,{\bar t})$, and singular at ${\bar t}$, it follows
from ~\ref{vlug} and the definition of $\Delta(t)$, that $M(t)$ is
non-singular (and, in fact, positive definite as shown in the previous section)
on $(0,{\bar t})$, and singular at ${\bar t}$ as well. Therefore, since $M(t)$
is symmetric and real-analytic, it follows from \cite{Kato1995} that we can find
a positive and real-analytic scalar function $\epsilon(t)$, and a real-analytic
unit vector $u(t)$ such that:
$$
M(t)u(t)=\epsilon(t) u(t),\; \epsilon(t)>0 
$$
on $(0,{\bar t})$, $\epsilon({\bar t})=0,\; |u(t)|_2=1$, where $|.|_2$ denotes
the Euclidean norm. In particular,
$M({\bar t})u({\bar t})=0$, and since $\Delta({\bar t})$ is non-singular, it
follows that $P({\bar t})u({\bar t})=0$. Then ~\ref{claim} implies that
$Q({\bar t})u({\bar t})\neq 0$. Define the real-analytic unit vector
$$
v(t)=\frac{\Delta(t)u(t)}{|\Delta(t)u(t)|_2},\;\; t\in (0,{\bar t}),
$$
and calculate
\begin{align*}
\lim_{t\rightarrow {\bar t}}|{\tilde S}(t)v(t)|_2&=\lim_{t\rightarrow {\bar t}}|Q(t)P^{-1}(t)v(t)|_2 \\
%&=\lim_{t\rightarrow {\bar t}}\frac{|Q(t)M^{-1}(t)\Delta^{-1}(t)(\Delta(t)u(t))|_2}{|\Delta(t)u(t)|_2}\\
&=\frac{|Q({\bar t})u({\bar t})|_2}{|\Delta({\bar t})u({\bar t})|_2}\lim_{t\rightarrow {\bar t}}\frac{1}{\epsilon(t)}=+\infty.
\end{align*}
Since for any real $n\times n$ matrix $A$, and for any unit vector $x$ (i.e. $|x|_2=1$) holds that 
$|Ax|_2\leq |A|_F$, it follows that $\lim_{t\rightarrow {\bar t}}|{\tilde
S}(t)|_F=+\infty$. This yields the sought-after contradiction.

By combining Proposition $\ref{main-proposition}$ and the results in this subsection, we have proved the main result concerning the generic emergence of balance for solutions of system ~\ref{flow}.
\begin{theorem}
There exists a dense set of initial conditions $X_0$ in $\reals^{n\times n}$ such that the corresponding solution 
$X(t)$ of~\ref{flow} satisfies: 
$$
\lim_{t\rightarrow {\bar t}}\frac{X(t)}{|X(t)|_F}=\frac{ww^T}{|ww^T|_F}.
$$
with $w = e^{\bar{t} A_0} V \Delta^{-1}({\bar t})u$.
\label{thm:generic}
\end{theorem}
\begin{proof}
The set of initial conditions $X_0$ for which $A_0 \neq 0$ and assumption
\ref{assum} holds is dense in $\reals^{n\times n}$.
\end{proof}

% Do NOT remove this, even if you are not including acknowledgments

\end{document}